\newif\ifDRAFT 
\DRAFTfalse
\DRAFTtrue 
\newif\ifSHOWCOMMENTS
\SHOWCOMMENTStrue
\SHOWCOMMENTSfalse 

\documentclass[11pt]{article}
\usepackage[fleqn]{amsmath}
\usepackage[margin=1in]{geometry}
\usepackage{graphicx} 

\usepackage{amsfonts, amsmath, amssymb, amsthm}

\usepackage{newtxtext}
\usepackage[varvw]{newtxmath}
\usepackage{enumitem}
\setlist[itemize]{itemsep=0pt, topsep=2pt}
\setlist[enumerate]{itemsep=0pt, topsep=2pt}
\setlist[description]{itemsep=0pt, topsep=2pt}

\usepackage{mathrsfs}  
\usepackage{algorithm} 
\usepackage[noend]{algpseudocode}
\usepackage{hyperref}
\usepackage{color}
\usepackage{array}
\usepackage{cleveref}
\usepackage{multirow}
\usepackage{tabularx}
\usepackage[table,xcdraw]{xcolor}
\definecolor{DarkRed}{rgb}{0.5,0.1,0.1}
\definecolor{DarkBlue}{rgb}{0.1,0.1,0.5}
\definecolor{ForestGreen}{rgb}{0.1333,0.5451,0.1333}
\hypersetup{colorlinks=true, linkcolor=DarkRed,citecolor=ForestGreen}
\usepackage[inkscapeformat=pdf]{svg}
\usepackage{multirow}
\usepackage{mathtools}
\usepackage{tabulary} 
\usepackage{booktabs}
\usepackage{arydshln}
\usepackage{pdfpages}
\usepackage{booktabs}
\usepackage[table,xcdraw]{xcolor}
\usepackage{arydshln}
\usepackage[x11names, table]{xcolor}
\usepackage{enumitem}
\usepackage{comment}
\usepackage{thm-restate}
\usepackage{tikz}
\usetikzlibrary{calc}
\usetikzlibrary{shapes.geometric,arrows.meta}

\theoremstyle{plain}
\newtheorem{theorem}{Theorem}[section]
\newtheorem{lemma}[theorem]{Lemma}

\newtheorem{observation}[theorem]{Observation}
\newtheorem{claim}[theorem]{Claim}
\newtheorem{fact}[theorem]{Fact}

\theoremstyle{definition}
\newtheorem{definition}[theorem]{Definition}

\providecommand{\set}[1]{{\{#1\}}}

\DeclareSymbolFont{bbold}{U}{bbold}{m}{n}
\DeclareSymbolFontAlphabet{\mathbbold}{bbold}

\newcommand{\cF}{\mathcal{F}}
\newcommand{\cT}{\mathcal{T}}
\newcommand{\cL}{\mathcal{L}}
\newcommand{\cN}{\mathcal{N}}
\newcommand{\cH}{\mathcal{H}}

\newcommand{\Att}{\mathcal{A}}

\newcommand{\Near}{\mathsf{Near}}
\newcommand{\Far}{\mathsf{Far}}

\newcommand{\LCA}{\mathrm{LCA}}

\newcommand{\Ltree}{\cL\cT}
\newcommand{\Children}{\mathcal{C}}

\newcommand{\polylog}{\operatorname{polylog}}

\algrenewcommand\algorithmicrequire{\textbf{Input:}}
\algrenewcommand\algorithmicensure{\textbf{Output:}}

\title{An Optimal Structure for All-Pairs Nearest Mincuts \\and Sensitivity Oracles for Edge Insertions}

\ifDRAFT
    \author{Koustav Bhanja \thanks{Weizmann Institute of Science, Israel. Email: \texttt{koustav.bhanja@weizmann.ac.il}. Supported by the Koshland Fellowship, and Merav Parter’s European Research Council (ERC) grant under the European Union’s Horizon 2020 research and innovation programme, grant agreement
No. 949083.}
    \and Yotam Kenneth-Mordoch  \thanks{Weizmann Institute of Science, Israel. Email: \texttt{yotam.kenneth@weizmann.ac.il}. }
    \and Asaf Petruschka \thanks{Weizmann Institute of Science, Israel. Email: \texttt{asaf.petruschka@weizmann.ac.il}. Supported by an Azrieli Foundation fellowship,
and by Merav Parter’s European Research Council (ERC) grant under the European Union’s Horizon 2020 research
and innovation programme, grant agreement No. 949083.}
    }
\else
\fi

\date{}

\begin{document}

\maketitle

\pagenumbering{roman}
\begin{abstract}
    Given an undirected weighted graph $G=(V,E)$ on $n$ vertices, the classical \textit{Gomory-Hu tree} of $G$ is a structure that encodes an arbitrary minimum $s,t$-cut for every $s,t\in V$ using just $O(n)$ space.
    In this work, we address whether the same compactness is achievable for the natural and structured family of all-pairs \textit{nearest minimum cuts}. The nearest minimum $(s,t)$-cut is the unique inclusion-wise minimal one among all minimum $s,t$-cuts containing $s$.  
    This family has proven useful in a wide range of applications, including fault-tolerant reachability, minimum cut sensitivity oracles, cactus representations, and fast Gomory-Hu tree constructions.
    Despite its fundamental role, no subquadratic-space representation is known for all-pairs nearest minimum cuts to date. 
    Representations using $O(n)$ space are known only in the restricted single-source setting, where given a source $s$ one can report the nearest minimum $(t,s)$-cut for any $t\in V\setminus \{s\}$. 

    We close this gap by presenting the first optimal-space representation of all-pairs nearest minimum cuts, providing a natural analogue of the Gomory-Hu tree.
    Our main result is an $O(n)$ space structure that encodes the nearest minimum cut between every pair of vertices. 
    Furthermore, given any pair $s,t\in V$, it can report the nearest minimum $s,t$-cut in $O(n)$ time. 
    Both bounds match those of the Gomory-Hu tree and are worst-case optimal.
    As an application, we leverage this data structure to design an all-pairs minimum cut \textit{sensitivity oracle} for edge insertion: 
    a data structure that occupies $O(n)$ space and, given a query edge $e$, can determine for all pairs $s,t\in V$ whether the minimum $s,t$-cut value increases upon insertion of $e$ in $O(n^2)$ total time, i.e.,\ amortized $O(1)$ time per vertex pair. 
    Existing insertion sensitivity oracles were either limited to the single-source setting or used $O(n^2)$ space for all-pairs [Baswana, Gupta, and Knollmann, Algorithmica'22; Baswana and Pandey, SODA'22]. 
    Technically, our results heavily utilize the sophisticated \textit{connectivity carcass} data structure for Steiner mincuts [Dinitz and Vainshtein, STOC'94], highlighting its usefulness in compactly representing minimum cuts.
\end{abstract}
\setcounter{page}{0}
\thispagestyle{empty}

\pagenumbering{arabic}
\newpage
{\small \setlength{\parskip}{0em}
\tableofcontents{}
}
\setcounter{page}{0}
\thispagestyle{empty}

\newpage
\pagenumbering{arabic}

\section{Introduction} 
The study of minimum cuts, in short \textit{mincuts}, is a fundamental theme in the graph algorithms literature. 
Let $G=(V,E)$ be an undirected graph on $n=|V|$ vertices and $m=|E|$ edges, with every edge having a positive weight.
A \textit{$(u,v)$-mincut}, for any $u,v\in V$, is a set of vertices $C\subset V$ that minimizes the total weight of edges between $C$ and $\overline{C}\coloneqq V\setminus C$ under the constraint $\set{u,v}\cap C =\set{   u}$.
A particular family of $(u,v)$-mincuts that has emerged as useful in a wide range of contexts is that of the \textit{nearest} (or conversely, \textit{farthest}) mincuts\footnote{Nearest mincuts are also referred to as the \textit{minimal}, \textit{latest}, \textit{tight}, or \textit{loose} cuts in the literature.}, which is defined as follows.%

\begin{definition} \label{def: nearest and farthest mincuts}
     The nearest  $(u,v)$-mincut, denoted by $\Near(u,v)$, is the unique $(u,v)$-mincut $C\subseteq V$ such that any $(u,v)$-mincut $C'\subseteq V$ satisfies $C\subseteq C'$. 
     The farthest $(u,v)$-mincut is defined as $\Far(u,v)\coloneqq \overline{\Near(v,u)}$. 
\end{definition}
Unlike an arbitrary $(u,v)$-mincut, $\Near(u,v)$ is uniquely determined and arises canonically in every maxflow computation.
As shown in the classical work of Ford and Fulkerson~\cite{ford_fulkerson_1956} on computing maxflow, $\Near(u, v)$ is precisely the set of vertices reachable from $u$ in the residual graph corresponding to any $u,v$-maxflow computation.
Since their inception, nearest mincuts have proven remarkably useful in a wide range of applications, including sparse fault-tolerant reachability subgraphs~\cite{BaswanaCR18}, sensitivity oracles for minimum cuts~\cite{DBLP:journals/mp/PicardQ80, BaswanaGK22, BhanjaPP2026, DBLP:conf/stoc/DinitzV94},  pseudo-deterministic global minimum cuts computations \cite{DBLP:conf/innovations/Agarwala026, DBLP:journals/corr/Pseudo-Det-MinCut-Yotam}, fast constructions of cactus for global mincuts~\cite{DBLP:conf/soda/KargerP09, HHS24} and Gomory-Hu trees (or all-pairs mincuts)~\cite{DBLP:conf/stoc/AbboudKT21,DBLP:conf/focs/AbboudKT21, DBLP:conf/focs/AbboudK0PST22, DBLP:conf/focs/AbboudKLPGSYY25}.
Given the importance of the nearest minimum cuts, it is natural to ask whether the family of all nearest minimum cuts, i.e. $\cN\coloneqq\set{\Near(u,v) \mid u,v\in V}$, admits a compact representation.
We give an affirmative answer to this question by showing that $\cN$ can be represented by a data structure of size $O(n)$, which achieves worst-case optimal space and query time. 

Designing compact representations for cut families is a central topic and appears in several algorithmic applications, namely dynamic algorithms, fault tolerance, edge connectivity augmentation, etc.
See \cite{dinitz1976structure, DBLP:journals/mp/PicardQ80, DBLP:conf/stoc/DinitzV94, DinitzN95, NagamochiI92, Karger00, DBLP:journals/jacm/KawarabayashiT19, edgeconnectivityaigmentation/naor1997fast, edgeconnaugmentation/cen2022augmenting, DBLP:conf/soda/AbboudKT22, DBLP:journals/talg/BaswanaBP23,  DBLP:conf/stoc/YKKrauthgamer26} and the references therein. 
The fundamental challenge is that these families often have an enormous size, e.g., $\Omega(n^2)$ global mincuts \cite{dinitz1976structure}, $\Omega(2^n)$ $(s,t)$-mincuts for any fixed $s,t\in V$. 
A landmark compact representation is the \emph{Gomory-Hu tree}, introduced in 1961 by Gomory and Hu~\cite{GomoryH61} and now a standard topic in algorithm textbooks. 
It encodes one arbitrary $(u,v)$-mincut for every pair of vertices $u,v \in V$ using only $O(n)$ space.
Moreover, given any query vertices $u$ and $v$, it can report a $(u,v)$-mincut in worst-case optimal $O(n)$ time. 

Our goal therefore can be framed as finding an analog of the Gomory-Hu tree that encodes $\cN$, i.e. all nearest minimum cuts.
An intriguing direction towards this end is the notion of a \emph{minimal Gomory-Hu tree} \cite{DBLP:conf/focs/AbboudK0PST22, DBLP:conf/focs/AbboudKLPGSYY25}, which arises in recent fast Gomory-Hu tree constructions. 
Informally, it is the unique Gomory-Hu tree of the graph $G_s$ obtained by adding an edge of sufficiently small weight from a fixed vertex $s$ to every other vertex $x\in V\setminus\{s\}$. 
This perturbation guarantees that for any pair of vertices $u,v\in V$, the only $(u,v)$-mincuts in $G_s$ are either $\Near(u,v)$ or $\Near(v,u)$, or both. 
Consequently, a minimal Gomory-Hu tree stores at least one of $\Near(u,v)$ and $\Near(v,u)$ for every $u,v\in V$, and hence at least half of $\cN$.
This immediately suggests the possibility that a small collection of such trees might suffice to represent all of $\cN$.
Unfortunately, we show that this intuition fails: encoding all the cuts in $\cN$ requires $\Omega(n)$ distinct minimal Gomory-Hu trees in the worst case (see \Cref{sec:nearest-gomory-hu-lower-bound}). 
Thus, any straightforward extension of the Gomory-Hu tree paradigm does not seem to succinctly represent $\cN$.
This makes it particularly interesting to explore whether $\cN$ has rich structural properties that admit a compact representation.

In the restricted single-source setting with a designated source vertex $s$, there exist compact representations for the family of nearest mincuts $\cN_s:=\{\Near(u,s)~|~u\in V\setminus\{s\}\}$ and farthest mincuts  $\cF_s:=\{\Far(u,s)~|~u\in V\setminus \{s\}\}$. 
Note that $\Near(u,s)$ may not necessarily be the same as $\Far(u,s)=\overline{\Near(s,u)}$, and thus, $\cN_s\ne \cF_s$; a simple example would be a clique on $n\ge 3$ vertices, where for every $u,v\in V$ we have $\Near(u,v)=\set{u}$ and $\Far(u,v)=V\setminus\set{v}$.
In a manner reminiscent of the Gomory-Hu tree, both families admit extremely compact $O(n)$ space representations that also support efficient $O(n)$ time queries.
The family of nearest mincuts to $s$, denoted by $\cN_s$, forms a laminar family by a simple application of submodularity of cuts, see \cite{HariharanKPB07,  DBLP:conf/stoc/AbboudKT21, BaswanaGK22}. 
Therefore, $\cN_s$ can be represented by a rooted tree $T_s$, known as the \textit{nearest mincut tree} of $s$, on $O(n)$ nodes:
Every vertex is mapped to a unique node in $T_s$, and $\Near(u,s)$ for any vertex $u\in V\setminus \{s\}$ is defined by the vertices mapped to the subtree rooted at the node containing $u$ in $T_s$. 
Meanwhile, the structure of $\cF_s$ is a bit more involved.
Recently, Baswana, Gupta, and Knollmann \cite{BaswanaGK22} have shown that  $\cF_s$ may not be laminar but still has similar properties that allow it to be represented by a DAG of size $O(n)$, called the \textit{farthest mincut DAG} of $s$.  The natural next question is whether these compact representations can be generalized to the all-pairs setting. 

\paragraph{Optimal space structure for all-pairs nearest mincuts.} 
Our main result is the following new data structure, called \textit{Nearest Mincut Hierarchy}, for storing all-pairs nearest mincuts in undirected graphs. 
\begin{theorem} \label{thm:main}
    For any undirected weighted graph $G=(V,E)$ on $n=|V|$ vertices, there is a  structure occupying $O(n)$ space that encodes $\Near(u,v)$ for every $u,v\in V$. Furthermore,  given any pair of vertices $(u,v)\in V\times V$ as a query, it can report the cut $\Near(u,v)$ in $G$ in worst-case $O(n)$ time.
\end{theorem}
\Cref{thm:main} generalizes the Gomory-Hu tree by storing $\Near(u,v)$, instead of an arbitrary $(u,v)$-mincut, while preserving the same $O(n)$ space and query time bounds.
It also combines and extends existing structures for nearest and farthest mincuts for a fixed source~\cite{HariharanKPB07,BaswanaGK22} into a unified framework, while matching their space and query time bounds.
Finally, note that there exists an information-theoretic lower bound of $\Omega(n\log{n})$ bits of space to store a Gomory-Hu tree which immediately implies the same lower bound for our result.%
\footnote{The lower bound follows from the fact that there are $n^{n-2}$ distinct labeled trees on $n$ vertices, and each such tree is a Gomory-Hu tree for itself but no other tree.}
Hence, the space complexity of the nearest mincut hierarchy is optimal.
Meanwhile, in directed unweighted graphs, using the result of \cite{Benczur95a}, it is simple to show that storing a $(u,v)$-mincut for every $u,v\in V$ requires  $\Omega(n^2)$ bits of space through an information-theoretic lower bound.
A comparison of our encoding result with existing ones is presented in \Cref{tab:encoding-comparison}. 

\begin{table}[!htbp]
\centering
\footnotesize
\setlength{\tabcolsep}{4pt}
\renewcommand{\arraystretch}{1.18}
\begin{tabular}{
p{0.385\textwidth}
p{0.12\textwidth}
p{0.20\textwidth}
p{0.10\textwidth}
p{0.11\textwidth}
}
\hline
\centering\textbf{Result}
& \centering\textbf{Scope}
& \centering\textbf{Cut type}
& \centering\textbf{Space}
& \centering\textbf{Query}
\tabularnewline
\hline

\centering Gomory-Hu tree~\cite{GomoryH61}
& \centering All pairs
& \centering Arbitrary mincut
& \centering $O(n)$
& \centering $O(n)$
\tabularnewline

\centering Nearest mincut tree~\cite{HariharanKPB07, DBLP:conf/stoc/AbboudKT21, BaswanaGK22}
& \centering Single source
& \centering Nearest
& \centering $O(n)$
& \centering $O(n)$
\tabularnewline

\centering Farthest mincut DAG~\cite{BaswanaGK22}
& \centering Single source
& \centering Farthest
& \centering $O(n)$
& \centering $O(n)$
\tabularnewline

\hline

\centering \textbf{This work}
& \centering \textbf{All pairs}
& \centering \textbf{Nearest (and farthest)}
& \centering $O(n)$
& \centering $O(n)$
\tabularnewline

\hline
\end{tabular}
\caption{Comparison of compact minimum cut families encodings. Here $n=|V|$ for the input graph $G=(V,E)$, and all graphs are weighted and undirected.}
\label{tab:encoding-comparison}
\raggedbottom
\end{table}

\paragraph{Mincut sensitivity oracles for the insertion of an edge.}
A sensitivity oracle for a graph problem is a compact data structure that efficiently reports whether the solution of the problem changes upon insertion or failure of edges/vertices.  
Such sensitivity oracles are crucial for monitoring large-scale real-world networks that undergo frequent changes, such as link additions/failures, network upgrades, and other structural changes that can affect performance.
Directly computing how such changes affect the networks is often unfeasible due to their large size.
Hence, sensitivity oracles using space much smaller than the graph itself, i.e., $o(n^2)$ space or even $O(n)$ space, are of more interest.

Initiated by Picard and Queyranne \cite{DBLP:journals/mp/PicardQ80}, the design of sensitivity oracles for mincuts is an emerging field of research \cite{BaswanaGK22, DBLP:conf/soda/BaswanaP22, DBLP:journals/talg/BaswanaBP23, DBLP:conf/icalp/BaswanaB24, DBLP:conf/esa/BaswanaBR25, DBLP:conf/icalp/Bhanja25, DBLP:conf/innovations/AhiCPPS26}.
Here we focus on the relevant results that are for undirected graphs. 
In general, the oracles handling insertions are considered comparatively easier than the ones handling failures. 
However, insertion oracles often require deep insights into the problem and act as a first step towards handling the failure case, e.g., see \cite{BaswanaGK22, BhanjaPP2026, DBLP:conf/icalp/Bhanja25, DBLP:conf/mfcs/BaswanaGT19}. 
Here, we study the insertion regime for the all-pairs mincut, which has received significant attention in the past few years. 
The authors in \cite{DBLP:journals/mp/PicardQ80} gave a sensitivity oracle for a fixed source-sink $s,t\in V$ and established the following fundamental connection between insertion sensitivity and nearest mincuts:
\\
\textit{\text{~~~~}For any $u,v\in V$, $(u,v)$-mincut value increases iff the inserted edge lies between $\Near(u,v)$ and $\Near(v,u)$.}
\\
Combining this observation with \Cref{thm:main} immediately yields the following new sensitivity oracle for all-pairs mincut occupying the worst-case optimal space and achieving $O(n)$ query time.
\begin{theorem}
    \label{cor:all-pairs-insertion-sensitivity-oracle}
    For any undirected weighted graph $G=(V,E)$ on $n=|V|$ vertices, there exists a data structure occupying $O(n)$ space such that, given a query containing an edge $e\in V\times V$ of any arbitrary weight and a pair of vertices $(u,v)\in V\times V$, it can determine in $O(n)$ time whether the $(u,v)$-mincut increases upon insertion of $e$ in $G$.
\end{theorem}
Despite significant progress in recent years, it remained unclear whether there exists an $o(n^2)$ space data structure for answering insertion queries in the all-pairs regime, as noted in \cite{DBLP:conf/soda/BaswanaP22}. 
In the single-source regime, where queries were limited to a pair $u,s$ for a fixed source $s\in V$, Baswana, Gupta, and Knollmann \cite{BaswanaGK22} gave a data structure with $O(n)$ space and $O(n)$ query time.
Our data structure in \Cref{cor:all-pairs-insertion-sensitivity-oracle} generalizes this result to the all-pairs regime, while maintaining the same space and query time guarantees.
Existing All-pairs data structures use $O(n^2)$ space, although they support $O(1)$ query time \cite{BaswanaGK22,DBLP:conf/soda/BaswanaP22}.
Additionally, there exist data structures for unweighted graphs using $O(\min(n^{1.5},m))$ space and $O(\min(n^{1.5},m))$ query time \cite{DBLP:conf/stoc/YKKrauthgamer26,DBLP:conf/soda/BaswanaP22}, but they do not seem to extend to the weighted case.

Furthermore, our encoding results cannot be strengthened due to the following lower bounds.
First, any data structure handling even two edge insertions (of unit capacity), even for the fixed source-sink setting, must occupy information-theoretically $\Omega(n^2)$ bits of space \cite{DBLP:conf/icalp/Bhanja25}. 
In addition, a known lower bound of \cite{BaswanaGK22} (see Theorem 5) shows that reporting the new capacity of $(u, v)$-mincut when the query edge can have weight polynomial in $n$ requires $\Omega(n^2\log~n)$ bits of space in the worst case, irrespective of the query time. 
This lower bound holds even in the restricted single-source setting.

A natural extension for sensitivity oracles is to report, upon insertion of a query edge, every pair $u,v\in V$ for which the $(u,v)$-mincut changed.
For single-source mincuts, the oracle of \cite{BaswanaGK22} can report in $O(n)$ time every $v\in V\setminus \{s\}$ such that the $(s,v)$-mincut is changed upon insertion of an edge.
Similarly, for all-pairs mincuts, the $O(n^2)$ space sensitivity oracles of \cite{BaswanaGK22, DBLP:conf/soda/BaswanaP22} can report also all $k$ pairs $u,v\in V$ such that $(u,v)$-mincut is changed in $O(k)$ time.
Observe that our data structure in \Cref{cor:all-pairs-insertion-sensitivity-oracle} can report all changed vertex pairs trivially in $O(n^3)$ time. 
However, we show that our data structure can be used to achieve the  worst-case optimal $O(n^2)$ time, as formally stated in the following theorem.
A comparison of our work with previous results is given in \Cref{tab:sensitivity-comparison}.
\begin{theorem}\label{thm : mincut sensitivity data structures for the insertion of an edge}
    For any undirected weighted graph $G=(V,E)$ on $n=|V|$ vertices, there is a data structure occupying $O(n)$ space such that, given a query containing an edge $e\in V\times V$ of any arbitrary weight, it can determine in amortized $O(1)$ time per vertex pair $(u,v)\in V\times V$ whether $(u,v)$-mincut is increased upon insertion of $e$ in $G$.     
\end{theorem}

To achieve \Cref{thm : mincut sensitivity data structures for the insertion of an edge}, we show that our data structure can be used to report, given a vertex $s$, the nearest mincut tree $T_s$ of $s$ in $O(n)$ time.
Observe that this is a significant improvement over \Cref{cor:all-pairs-insertion-sensitivity-oracle}, as it essentially finds the $\Near(v,s)$ for every $v\in V$ in just $O(n)$ time, which is again worst-case optimal. 
We state this result below, and believe that it might be of independent interest.
\begin{theorem} \label{thm : nearest mincut tree reporting}
    For any undirected weighted graph $G=(V,E)$ on $n=|V|$ vertices, there is a data structure occupying $O(n)$ space such that, given a query containing a vertex $s$, it can report the nearest mincut tree of $s$ in $O(n)$ time. 
\end{theorem}
\begin{table}[ht]
\centering
\footnotesize
\setlength{\tabcolsep}{3.2pt}
\renewcommand{\arraystretch}{1.18}
\begin{tabular}{
p{0.27\textwidth}
p{0.11\textwidth}
p{0.11\textwidth}
p{0.12\textwidth}
p{0.16\textwidth}
p{0.14\textwidth}
}
\hline
\centering\textbf{Result}
& \centering\textbf{Scope}
& \centering\textbf{Graph}
& \centering\textbf{Space}
& \centering\textbf{Single-pair query}
& \centering\textbf{All affected pairs}
\tabularnewline
\hline

\centering Single-source insertion oracle~\cite{BaswanaGK22}
& \centering Single source
& \centering Weighted
& \centering $O(n)$
& \centering $O(n)$
& \centering $O(n)$
\tabularnewline

\centering Single-source insertion oracle faster query~\cite{BhanjaPP2026}
& \centering Single source
& \centering Weighted
& \centering $O(n^{3/2})$
& \centering $O(1)$
& \centering $O(k)$
\tabularnewline

\centering All-pairs insertion oracle~\cite{BaswanaGK22,DBLP:conf/soda/BaswanaP22}
& \centering All pairs
& \centering Weighted
& \centering $O(n^2)$
& \centering $O(1)$
& \centering $O(k)$
\tabularnewline

\centering Unweighted insertion oracle~\cite{DBLP:conf/soda/BaswanaP22, DBLP:conf/stoc/YKKrauthgamer26}
& \centering All pairs
& \centering Unweighted
& \centering $O(\min\{n^{3/2},m\})$
& \centering $O(\min\{n^{3/2},m\})$
& \centering $O(n^2)$
\tabularnewline

\hline

\centering \textbf{This work}
& \centering \textbf{All pairs}
& \centering \textbf{Weighted}
& \centering $O(n)$
& \centering $O(n)$
& \centering $O(n^2)$
\tabularnewline

\hline
\end{tabular}
\caption{
Comparison of sensitivity oracles for mincuts under insertion of an edge.
Here $n=|V|$, $m=|E|$, and $k$ is the number of affected pairs.
}
\label{tab:sensitivity-comparison}
\end{table}


\section{Preliminaries}
\label{sec: preliminaries}
Throughout, we consider an undirected, positively edge-weighted and connected graph $G = (V,E, w)$ where $w:E\mapsto\mathbb{R}_+$.
We also denote the number of vertices by $n = |V|$ and the number of edges by $m = |E|$.

For $A \subseteq V$, its capacity $c(A)$ is the sum of weights $w(e)$ over all edges $e \in E$ such that $e$ goes between $A$ to $\overline{A} := V \setminus A$;
those edges are called the \emph{edge-set} or the 
\emph{contributing edges} of $A$.
If $\emptyset \neq X \subseteq A$ and $\emptyset \neq Y \subseteq \overline{A}$, we say that $A$ is an \emph{$(X,Y)$-cut}.
An $(X,Y)$-cut of smallest capacity is called \emph{$(X,Y)$-mincut}.
The term $X,Y$-(min)cut refers to either an $(X,Y)$-(min)cut or a $(Y,X)$-(min)cut.
The capacity of $X,Y$-mincut is denoted by $\lambda_{X,Y}$.
When $X,Y$ are singletons $X = \{x\}$ and $Y = \{y\}$ we may write $x,y$-(min)cut, and such shorthands apply to other notations as well.
For $S \subseteq V$, the cut (defined by) $A$ is said to be an \emph{$S$-cut} or a \emph{Steiner cut for $S$} if $S \cap A \neq \emptyset$ and $S \cap \overline{A} \neq \emptyset$.
An $S$-cut of smallest capacity is called \emph{$S$-mincut} or \emph{Steiner mincut for $S$}, and its capacity is denoted by $\lambda_S$.

It is well known that capacities satisfy \emph{sub/posimodularity}: for every $A,B \subseteq V$,
$c(A) + c(B) \geq c(A \cap B) + c(A \cup B)$ (submodularity), and $c(A) + c(B) \geq c(A \setminus B) + c(B \setminus A)$ (posimodularity).
The following lemma is a consequence that will be used many times rather than applying these directly (proof is in Appendix \ref{sec: missing proofs}):
\begin{lemma}\label{lem:sub-posi-general}
    Let $S,T \subseteq V$, and suppose $A$ is an $S$-mincut and $B$ is a $T$-mincut.
    \begin{itemize}
        \item (``Submodularity'') If one of $A \cap B$ and $A \cup B$ is an $S$-cut and the other is a $T$-cut, then the former is an $S$-mincut and the latter is a $T$-mincut.
        
        \item (``Posimodularity'') If one of $A \setminus B$ and $B \setminus A$ is an $S$-cut and the other is a $T$-cut, then the former is an $S$-mincut, the latter is a $T$-mincut
    \end{itemize}
    Furthermore, the 
    above also holds if we replace ``$S$-(min)cut'' by ``$S_1,S_2$-(min)cut'', and/or replace ``$T$-(min)cut'' by ``$T_1,T_2$-(min)cut'', for  $S_1, S_2, T_1, T_2 \subseteq V$. 
\end{lemma}

By submodularity, we immediately see that if $A$ and $B$ are $(X,Y)$-mincuts, then so are $A \cap B$ and $A \cup B$, which shows that the notions of nearest/farthest from~\Cref{def: nearest and farthest mincuts} mincuts are well-defined.
 We additionally have the following ``sub/posimodularity-like'' inequalities for undirected graphs:
\begin{lemma} [Four component Lemma, \protect{cf.\ \cite[Problem 6.48(iii)]{lovasz1979combinatorial}}] \label{lem: four component lemma}
    For any three sets $A,B,C\subseteq V$, $ c(A)+c(B)+c(C) \ge c(A\setminus (B\cup C))+c(B\setminus (A\cup C))+c(C\setminus (A\cup B))+c(A\cap B\cap C) $
\end{lemma}

\subsection*{Key Elements of the Connectivity Carcass}

We now introduce several key aspects of the Connectivity Carcass from the seminal works of Dinitz and Vainshtein~\cite{DBLP:conf/stoc/DinitzV94, DBLP:conf/soda/DinitzV95, DBLP:journals/siamcomp/DinitzV00}.
The carcass is a complicated and multifaceted data structure for Steiner mincuts.
Here, we only give the ``bare minimum'' interface needed for our purposes.
We also refer to~\cite{BaswanaP25} for an excellent and simple presentation of the carcass, including full proofs of the results presented here. 

Let $S \subseteq V$ be a set of vertices called \emph{terminals}.
Two terminals $s,t \in S$ are called \textit{$S$-equivalent} if there is no $S$-mincut that separates $s$ and $t$.
We denote the number of $S$-equivalence classes by $d(S)$.
A partition of $S$ into $S_1,S_2$ is called a \emph{valid partition} if there is some $S$-mincut $C$ s.t.\ $C \cap S = S_1$ and $\overline{C} \cap S = S_2$ (i.e., if every $S_1,S_2$-mincut is an $S$-mincut).
For a vertex $v \in V$, we say that a valid partition $S_1,S_2$ \emph{distinguishes} $v$ if there exist two $S_1,S_2$-mincuts, one keeping $v$ on the side of $S_1$, and the other keeping $v$ on the side of $S_2$.
A valid partition $S_1,S_2$ is said to be \emph{laminar} if for every other valid partition $S'_1, S'_2$, there is at least one containment among the four sets $S_1,S_2,S'_1,S'_2$.

We are now ready to give the succinct carcass interface required for this paper, summarized in the following two theorems:

\begin{theorem}[Skeleton]\label{thm:skeleton}
    There is a tree $\cH_S$ with $O(|d(S)|)$ nodes and a mapping 
    $\pi_S (\cdot)$ from $S$ to nodes of $\cH_S$, called the \emph{skeleton} of $S$, 
    whose edges \emph{bijectively} correspond to laminar valid partitions of $S$:
    edge $g$ of $\cH_S$ corresponds to the laminar valid partition $S_1,S_2$ s.t.\ $S_1$ are the terminals mapped to one connected component of $\cH_S \setminus g$ (i.e., $\cH_S$ with $g$ deleted), 
    and $S_2$ are those mapped to the other component. The mapping $\pi_S (\cdot)$ need not be surjective, i.e., there could be nodes of $\cH_S$ without mapped terminals, called empty nodes.
\end{theorem}


\begin{theorem}[Projections]\label{thm: proj vertex}
    Each vertex $v \in V$ has a projection $\pi_S (v)$ in $\cH_S$ (which coincides with the previous definition from~\Cref{thm:skeleton} if $v \in S$), which is a path in the skeleton tree $\cH_S$ (possibly a degenerate path consisting of just a single node), such that the following properties hold:
    \begin{itemize}
        \item The skeleton edges in $\pi_S (v)$ 
        correspond precisely to those laminar valid partitions that distinguish $v$. 

        \item Let $g$ be an edge of $\cH_S$, corresponding to the laminar valid partition $S_1,S_2$.
        If $\pi_S (v)$ is fully contained in the same connected component as $S_1$ in $\cH_S \setminus g$, then every $S_1,S_2$-mincut keeps $v$ with $S_1$.
    \end{itemize}
\end{theorem}

The following lemma gives a key connection between laminar valid partitions and nearest mincuts:
\begin{lemma}\label{lem:laminar valid partition}
    Let $s,t \in S$ be two terminals that are not $S$-equivalent.
    Then the partition of $S$ defined by $\Near(s,t)$ is a laminar valid partition.
    Furthermore, in the skeleton $\cH_S$, this partition corresponds to the unique skeleton edge incident on the node $\pi_S (s)$ which is in the direction towards the node $\pi_S (t)$.
\end{lemma}

\section{Technical Overview}
\label{sec:technical-overview}

We now give an overview of our main data structure, stated in~\Cref{thm:main}, that given query $(x,y) \in V \times V$, can report the nearest mincut $\Near(x,y)$.
For clarity, we focus in the overview on the $O(n)$ space aspect of the data structure and disregard the query time.
The construction of our sensitivity oracles for edge insertion (\Cref{thm : mincut sensitivity data structures for the insertion of an edge}), and the query for reporting a nearest mincut tree for a given source vertex (\Cref{thm : nearest mincut tree reporting}) leverage the same insights, although requiring some further technical complications.

\subsection{The Nearest Mincut Hierarchy} 
\label{sec: skeleton hierarchy}
In this section, we introduce our data structure, the \emph{nearest mincut hierarchy}, for compactly storing all nearest mincuts of $G$. 
The construction of the nearest mincut hierarchy is motivated by an $O(n^2)$ space sensitivity oracle of \cite{DBLP:conf/soda/BaswanaP22}. 
The main structure underlying  their oracle is a combination of the hierarchy tree for mincuts of \cite{KatzKKP04} and the Skeleton of Connectivity Carcass \cite{DBLP:conf/stoc/DinitzV94}. 
This structure, which we refer to as the \textit{Skeleton hierarchy} and describe below, serves as the starting point of our construction. 

\paragraph{Skeleton Hierarchy.} 
The first component, the hierarchy tree $\cT$ of \cite{KatzKKP04}, is defined recursively as follows.
The root of $\cT$ is a node that corresponds to the entire vertex set $V$ (not stored explicitly).
Then, for every node $\mu\in \cT$, its children are the $\mu$-equivalence classes. 
Notice that the tree $\cT$ has exactly $n$ leaves, each corresponding to a vertex $v\in V$, and that each non-leaf node in $\cT$ has at least two children.
Therefore, the total number of nodes in $\cT$ is $O(n)$, and it can be stored in $O(n)$ space. 
For brevity, we use $\mu$ to denote both the node and its corresponding vertex subset.
The \emph{skeleton hierarchy} is then obtained by augmenting each non-leaf node $\mu$ of $\cT$ with a \emph{skeleton} $\cH_\mu$ from \Cref{thm:skeleton} for terminal set $\mu$.
We shall use the following \textit{fact} regarding $\cT$:
\begin{fact}
    \label{fact:lca-mincut-value}
    Let $u,v\in V$ (which are mapped to leaves in $\cT$), the $\LCA_{\cT}(u,v)$-mincut value is the same as $u,v$-mincut value, where $\LCA_{\cT}(u,v)$ is the lowest common ancestor of $u,v$ in $\cT$. 
\end{fact}

Note that the authors in \cite{DBLP:conf/soda/BaswanaP22} use the most general \textit{cactus} version of the Skeleton, whereas our construction, focusing on nearest mincuts, relies only on a simplified tree variant (given in \Cref{thm:skeleton}), resulting in a conceptually simpler structure.
We now argue that the total space used by all the skeletons is also $O(n)$.
Fix some node $\mu\in \cT$.
The mapping $\pi_\mu (\cdot)$ of~\Cref{thm:skeleton} assigns all vertices of a given $\mu$-equivalence class to the same node $\cH_\mu$.%
\footnote{By~\Cref{thm:skeleton}, $u,v \in \mu$ are mapped to different nodes of $\cH_\mu$ if and only if there is a laminar valid partition of $\mu$ that separates them. 
If $u,v$ are $\mu$-equivalent then no such partition exists; otherwise, the partition of $\mu$ by $\Near(u,v)$ is such by~\Cref{lem:laminar valid partition}.}
Thus, the non-empty nodes of $\cH_{\mu}$ are exactly the children of $\mu$ in $\cT$.
Using \Cref{thm:skeleton}, it follows that the size of $\cH_\mu$, including empty nodes, is $O(d_\cT(\mu))$, where $d_\cT(\mu)$ is the degree of $\mu$ in $\cT$ (which is its number of children plus one).
Therefore, denoting the set of internal nodes in $\cT$ by $\cT_{int}$, the total size of all skeletons is $\sum_{\mu\in \cT_{int}} |\cH_{\mu}| = \sum_{\mu\in \cT_{int}} d_{\cT}(\mu) = O(n)$, where the second equality is since $\cT$ has $O(n)$ nodes.

\paragraph{A high-level view of our data structure.} 
The Skeleton hierarchy alone captures only a \textit{small} portion of $\Near(x,y)$. 
It can reveal only how $\Near(x,y)$  partitions the $\LCA_{\cT}(x,y)$.
However, as $\LCA_{\cT}(x,y)$ is usually a small subset of $V$, this information is far from sufficient to recover $\Near(x,y)$.
We illustrate this next in the base case (\Cref{sec : base case}). 
One simple approach to report all of $\Near(x,y)$ is to store, for every vertex, its projection onto every augmented skeleton in the Skeleton hierarchy as was done in \cite{DBLP:conf/soda/BaswanaP22}. 
Unfortunately, this uses $O(n^2)$ space in the worst case.

Our main contribution is reporting $\Near(x,y)$, while preserving the $O(n)$ space bound, by augmenting the Skeleton hierarchy with the following major components. 
\begin{enumerate}
    \item A carefully chosen set of only $O(n)$ vertex projections, each requiring $O(1)$ space. 
    These projections resolve the majority of non-trivial cases arising in our query algorithms.
    \item A new tool for nearest mincuts, called \textit{chain maximizers}: ordered vertex pairs whose nearest mincuts are containment-maximal in regard to a certain cut family, which is introduced in \Cref{sec: overview generic step}. 
    These are the key for handling the remaining difficult cases in our query algorithms.  
    Specifically, two chain maximizers are stored alongside each stored vertex projection, with each requiring only $O(1)$ space.
\end{enumerate}
While this data structure is sufficient for reporting $\Near(x,y)$, to improve the query time to $O(n)$, we also augment it with standard tree data structures for lowest common ancestors ($\LCA$) and level ancestor (LA) queries. 
The \textit{nearest mincut hierarchy} is the data structure that consists of all the above components.
Its space bound is summarized in the following lemma, whose proof is given in \Cref{sec: data structure}.
\begin{lemma} \label{lem: nearest mincut hierarchy space}
    The nearest mincut hierarchy of $G$ occupies $O(n)$ space. 
\end{lemma}

Our goal for the rest of the technical overview is to give intuition on how to use the nearest mincut hierarchy to find $\Near(x,y)$ for any query $(x,y) \in V \times V$. 
Along the way, we introduce the above two augmentations to the Skeleton hierarchy, and explain why they preserve the overall $O(n)$ space bound.
This is summarized in the following lemma, whose proof is provided in \Cref{sec:near-minimum-cut-query}. 
It is clear that \Cref{thm:main} is established immediately by combining \Cref{lem: nearest mincut hierarchy space} with the following \Cref{lem: nearest mincut hierarchy query}.
\begin{lemma}
    \label{lem: nearest mincut hierarchy query}
    The nearest mincut hierarchy of $G$ can report $\Near(x,y)$ in $O(n)$ time for any given ordered pair of vertices $(x,y)\in V\times V$.  
\end{lemma}

\subsection{The Base Case : \texorpdfstring{$\LCA_\cT(x,y)$}{LCA(x,y)}} \label{sec : base case}
Let us start by presenting the most basic connection between the hierarchy tree $\cT$ of the Skeleton hierarchy to $\Near(x,y)$, which serves as the starting point for our query algorithm.
Consider the node $\mu = \LCA_{\cT}(x,y)$. 
Recall from \Cref{fact:lca-mincut-value} that the $\mu$-mincut value is the same as the $x,y$-mincut value. 
Hence, $x,y$ are in different $\mu$-equivalence classes, namely $\pi_\mu (x)$ and $\pi_\mu (y)$, both of which are nodes in $\cH_{\mu}$.
Thus, by~\Cref{lem:laminar valid partition}, we can find out how $\mu$ is partitioned by $\Near(x,y)$ from the skeleton $\cH_\mu$.
Letting $g$ be the skeleton edge incident on $\pi_\mu (x)$ in the direction of $\pi_\mu (y)$, we have that $\mu \cap \Near(x,y)$ consists of (the vertices of $\mu$ mapped to) the connected component of $\pi_\mu (x)$ in $\cH_\mu \setminus g$.

Our high-level strategy for the query algorithm is to ``build-up'' $\Near(x,y)$, starting from $\mu = \LCA_\cT (x,y)$ and going upwards along the path $(\mu = \alpha_1, \alpha_2, \dots, \alpha_k = V)$ to the root of $\cT$.
We start by determining $\Near(x,y) \cap \alpha_1$ as was just described.
Then, we work our way up in \emph{steps}: at each step $i= 2,3, \dots, k$, we expand our knowledge from $\Near(x,y) \cap \alpha_{i-1}$ to the superset $\Near(x,y) \cap \alpha_i$.
Thus, after the final step $k$ is completed, we know the entirety of $\Near(x,y)$.
While the skeleton hierarchy suffices for the initialization of the query algorithm (i.e., finding $\Near(x,y) \cap \alpha_1$),
executing the generic step from $\alpha_{i-1}$ to $\alpha_i$ is the main challenge.
To do so, we have to use the full information stored in the nearest mincut hierarchy and leverage delicate structural insights.

\subsection{The First Step : Parent of \texorpdfstring{$\LCA_\cT(x,y)$}{LCA(x,y)}} \label{sec: overview first step}
To gain intuition, we start by tackling the first step of the query process, going from $\mu = \alpha_1$ to its parent $\nu \coloneqq \alpha_2$.
Namely, our goal for now is extending $\Near(x,y) \cap \mu$ to $\Near(x,y) \cap \nu$.

\paragraph{A succinct use of projections.}
Consider any sibling node $\eta$ of $\mu$ in $\cT$ (i.e., $\eta$ also has $\nu$ as its parent in $\cT$).
Using submodularity, we obtain that $\eta$ is either entirely contained in $\Near(x,y)$ (``in'') or entirely disjoint from it (``out'') (the proof is given in Appendix \ref{sec: missing proofs} (\Cref{lem: mincut does not separate any incomparable node})).
Hence, it suffices to classify each sibling $\eta$ as in or out.
To this end, we need to utilize not only skeletons, but also \emph{projections} (see \Cref{thm: proj vertex}).

Assume we are given the projection $\pi_\mu (v)$ in the skeleton $\cH_\mu$ of an arbitrary $v \in \eta$.
Recall that $\Near(x,y) \cap \mu$ corresponds to a connected component of $\cH_\mu \setminus g$ as discussed before.
By (the second item) of~\Cref{thm: proj vertex}, $v \in \Near(x,y)$ if and only if $\pi_\mu (v)$ lies completely inside this connected component.
This, crucially, relies on the fact that $\mu$ is the LCA of $x,y$ in $\cT$, which implies that $\Near(x,y)$ is a $\mu$-mincut.
Therefore, to classify $\eta$, it is sufficient to store $\pi_\mu (v)$ for some arbitrary $v \in \eta$, which uses additional $O(1)$ space.%
\footnote{The projection $\pi_\mu (v)$ is uniquely determined by its endpoints in the skeleton tree $\cH_\mu$, which is already stored.} 
However, storing a pair of projections for every pair of siblings $\mu, \eta$, would use $\sum_{\nu} \binom{d_{\cT}(\nu)}{2}$ space, which might be $\Omega(n^2)$ and hence too expensive \footnote{Consider an even cycle $u_1,\ldots,u_n$ with alternating $2$ and $1$ edge weights, $w(u_1,u_2)=2$. The root of the hierarchy has children $\{u_{2i-1},u_{2i}\}$ for $i \in [1,n/2]$. If for each pair of them $\mu_1, \mu_2$, we store a projection from $\mu_1$ onto $\cH_{\mu_2}$, we would get $\Omega(n^2)$ space.}.
Instead, our data structure stores a small set of only $O(n)$ projections, which we show to suffice as follows.

Our first tool towards reducing the number of stored projections is a simple (yet crucial) structural lemma, which lets us ``compress'' several siblings $\eta$ together.
Consider the skeleton of the parent $\cH_\nu$, and recall that all children of $\nu$ are nodes in this skeleton.
We call the connected components of $\cH_\nu\setminus\mu$ the \emph{$\mu$-subtrees} of $\cH_\nu$, and observe that they partition all of $\mu$'s siblings among them.
The key insight is that any $\mu$-subtree can contain siblings of only one type: in or out.
This is formalized in the following lemma.
From now on, we slightly abuse notation and identify a $\mu$-subtree $A$ in $\cH_\nu$ also with the set of vertices in $\nu$ mapped to it (i.e., with the union of $\mu$'s siblings contained in $A$).

\begin{lemma}[Subtree classification]\label{lem: subtree classification}
    Let $A$ be a $\mu$-subtree of $\cH_\nu$, and consider two vertices $x,y \in \mu$.
    Then either $A \subseteq \Near(x,y)$ (``$A$ is an in-$\mu$-subtree'') or $A \cap \Near(x,y) = \emptyset$ (``$A$ is an out-$\mu$-subtree''). 
    The edge connecting $\mu$ with an in-$\mu$-subtree (resp., out-$\mu$-subtree) is called an in-edge (resp., out-edge) of $\mu$.
\end{lemma}
\begin{proof}
    Suppose $A \not\subseteq \Near(x,y)$.
    By the properties of $\cH_\nu$ (\Cref{thm:skeleton}), there is a $\nu$-mincut $C$ such that $\mu \subseteq C$ and $A \subseteq \overline{C}$.
    Then $\Near(x,y) \cup C$ is a $\nu$-cut (separating $\mu$ from some node in $A$), and $\Near(x,y) \cap C$ is an $(x,y)$-cut.
    Hence, by submodularity (\Cref{lem:sub-posi-general}), the latter is an $(x,y)$-mincut.
    We conclude that $\Near(x,y) \subseteq C$ by definition of nearest mincut, and hence $A \cap \Near(x,y) = \emptyset$.
\end{proof}

Note that the above lemma holds even if $\mu$ is not the LCA of $x,y$, but any node $\alpha_i$ containing both of $x,y$.
We use this fact later for extending the cut from $\alpha_{i}$ to $\alpha_{i+1}$ for general $i$.
In light of~\Cref{lem: subtree classification}, it is sufficient to store the following projections in our data structure: For every child $\mu$ of a node $\nu$, and every $\mu$-subtree $A$ of $\cH_\nu$, choose an arbitrary \emph{representative vertex} $v \in A$, and store $\pi_\mu (v)$.
The total space used is,
\begin{equation} \label{eq : size of all projections}
    \sum_{\nu \text{ non-leaf}} \, \sum_{\mu \text{ child of } \nu} \deg_{\cH_\nu} (\mu)
\leq
\sum_{\nu \text{ non-leaf}} O(d_{\cT}(\nu))
\leq O(n),
\end{equation}
where the first inequality is because $\cH_\nu$ is of size $O(d_{\cT}(\nu))$, and the second is by the previous discussion on the skeleton hierarchy. 
At query time, these projections let us determine which representative vertices of $\mu$-subtrees are inside $\Near(x,y)$.
Thus, we can classify all the $\mu$-subtrees (and hence all of $\mu$'s siblings) as in or out of $\Near(x,y)$.
This concludes the discussion of the first step, extending $\Near(x,y)$ from its intersection with the LCA of $x,y$ in $\cT$, to its intersection with the parent of this LCA.

\subsection{The Generic Steps : Any ancestor of \texorpdfstring{$\LCA_\cT(x,y)$}{LCA(x,y)}}
\label{sec: overview generic step}
We now provide an overview of the most generic step: extending $\Near(x,y)\cap\alpha_{i}$ to $\Near(x,y)\cap\alpha_{i+1}$, for some $i> 1$. 
Let $\mu=\alpha_i$ and its parent $\nu=\alpha_{i+1}$. 
Henceforth, let $\alpha=\pi_\mu(x)$, i.e. the node in $\cH_\mu$ containing $x$.

\paragraph{Characterization of Skeleton subtrees.} 
By \Cref{lem: subtree classification}, to retrieve $\Near(x,y)\cap \nu$, it suffices to classify each of the $\mu$-subtree in $\cH_\nu$.
Let $v$ be a representative vertex of some yet unclassified $\mu$-subtree  in $\cH_\nu$. 
Notice that we know which of the $\alpha$-subtrees of $\cH_\mu$ are ``in" and ``out" for $\Near(x,y)$, since we know $\Near(x,y)\cap \mu$.
Furthermore, we have $\pi_\mu(v)$, the stored projection of $v$ on $\cH_\mu$.
We now explore all the possible configurations that can arise for $\pi_\mu(v)$ (refer to \Cref{fig:configurations}) with respect to the in/out $\alpha$-subtrees.

\begin{enumerate}[label=(\textbf{\Roman*})]
     \item Both endpoints of $\pi_{\mu}(v)$ belong to in-$\alpha$-subtrees. 
     \item At least one endpoint of $\pi_\mu(v)$ belongs to an out-$\alpha$-subtree.
     \item One endpoint of $\pi_\mu(v)$ is $\alpha$ and the other endpoint belongs to an in-$\alpha$-subtree.
     \item  The projection $\pi_\mu(v)=\alpha$.
 \end{enumerate}
To build the intuition behind our approach, let us revisit the easy first step, when $\mu=\LCA_\cT(x,y)$, with this subtree classification perspective.
Here, $\Near(x,y)$ is itself a $\mu$-mincut, and thus, it induces a valid partition of $\mu$. 
It follows that exactly one of all the $\alpha$-subtrees is an ``out" subtree, namely the one containing $y$, and the remaining subtrees are in-$\alpha$-subtrees for $\Near(x,y)$.
Furthermore, $\alpha$ is completely contained in $\Near(x,y)$. 
In this setting, we easily classify $v\in \Near(x,y)$ in cases (I), (III), and (IV), as both endpoints (and hence the entire projection $\pi_{\mu}(v)$) are contained in $\Near(x,y)$.
Otherwise, $v\notin  {\Near(x,y)}$ for case (II) as $\pi_{\mu}(v)$ \textit{touches} ``out" subtree.
Both results follow immediately from \Cref{thm: proj vertex}. 
In summary, in the first step, a $\mu$-subtree can be classified as ``in" or ``out" based on only the projection of its representing vertex $v$ onto the known classification of $\alpha$-subtrees in skeleton $\cH_\mu$.
We now extend this classification framework to the generic step where $\mu$ is a strict ancestor of $\LCA_\cT(x,y)$. 

We refer the reader to \Cref{fig:genstep} for an illustration of the process. 
Recall that we do not have projection $\pi_{\LCA_{\cT}(x,y)}(v)$ due to space constraints. 
Instead, we rely on the projection $\pi_{\mu}(v)$ in $\cH_\mu$.
One immediate hurdle is that $\Near(x,y)$ has a strictly larger mincut value than $\mu$-mincut, and thus no longer induces a valid partition of $\mu$.
In particular, since $x,y\in \alpha$, it is straightforward to see that $\alpha$ is neither completely ``in" nor ``out" of $\Near(x,y)$. 
Since $\Near(x,y)$ is no longer a $\mu$-mincut, it is impossible to directly apply \Cref{thm: proj vertex} to determine whether a $\mu$-subtree of $\cH_\nu$ is in or out.
Fortunately, we are able to show that the intuition from the first step survives intact for configurations (I) and (II).
This requires a careful application of the submodularity property (and its extensions) on $\Near(x,y)$ and $\mu$-mincuts.
A taste of these results is provided in the following lemma for one of the cases from configuration (I).
\begin{lemma} \label{lem: fully contained in case}
    Suppose both endpoints of $\pi_\mu(v)$ are in the same in-$\alpha$-subtree. Then, $v\in \Near(x,y)$.
\end{lemma}
\begin{proof}
    Let $A$ be the $\alpha$-subtree that contains both endpoints of $\pi_\mu(v)$. 
    Since $\pi_\mu(v)$ is a path in $\cH_\mu$, this immediately implies $\pi_\mu(v)\subseteq A\subseteq \Near(x,y)$.
    Let $g$ be the Skeleton edge incident on $A$ and $\alpha$, which defines the valid partition $A,\mu\setminus A$. 
    Let $C$ be a corresponding $\mu$-mincut of this valid partition such that $A\subseteq C$ and $\mu\setminus A\subseteq \overline{C}$. 
    Hence, $\alpha\subseteq\overline{C}$. 
    Observe that $y\not\in C$ as $C\cap \mu=A$, $A\subseteq \Near(x,y)$ and $y$ is not in $\Near(x,y)$ by definition.
    It follows that $C\cap \Near(x,y)$ is a $\mu$-cut and $A\cup \Near(x,y)$ is an $x,y$-cut. 
    Hence, $C\cap \Near(x,y)$ is a $\mu$-mincut separating $A,\alpha$ by applying submodularity on the cuts $A$ and $\Near(x,y)$. 
    By \Cref{thm: proj vertex}, $v$ is included in every minimum cut separating $A$ from $\alpha$, and in particular $v\in A\cap \Near(x,y)$, which implies that $v\in \Near(x,y)$.
\end{proof}
Case (IV) is solved by storing the projection of $v$ in some descendant of $\mu$ in $\cT$, in which we can apply the same machinery to it as in cases (I-III).
Finally, we remain with the {\em more interesting} case (III), which turns out to be more challenging and requires additional technical machinery. 

In case (III), $\pi_\mu(v)$ has one endpoint in an in-$\alpha$-subtree, say $T_{in}$, and its other endpoint is $\alpha$. 
Unlike cases (I) and (II), as discussed, the available information seems insufficient to determine whether $v\in \Near(x,y)$ since $\Near(x,y)\cap \alpha$ and $\overline{\Near(x,y)}\cap \alpha$ are both nonempty. 
It is indeed the case that, as shown in \Cref{fig : Chain tool}, there may exist two vertex pairs $x_1,y_1$ and $x_2,y_2$ in $\alpha$ such that $T_{in}\subseteq \Near(x_1,y_1)\cap \Near(x_2,y_2)$, yet $v\notin \Near(x_1,y_1)$ while $v\in \Near(x_2,y_2)$.
\begin{figure}
  \begin{center}
    \includegraphics[width=0.9\textwidth]{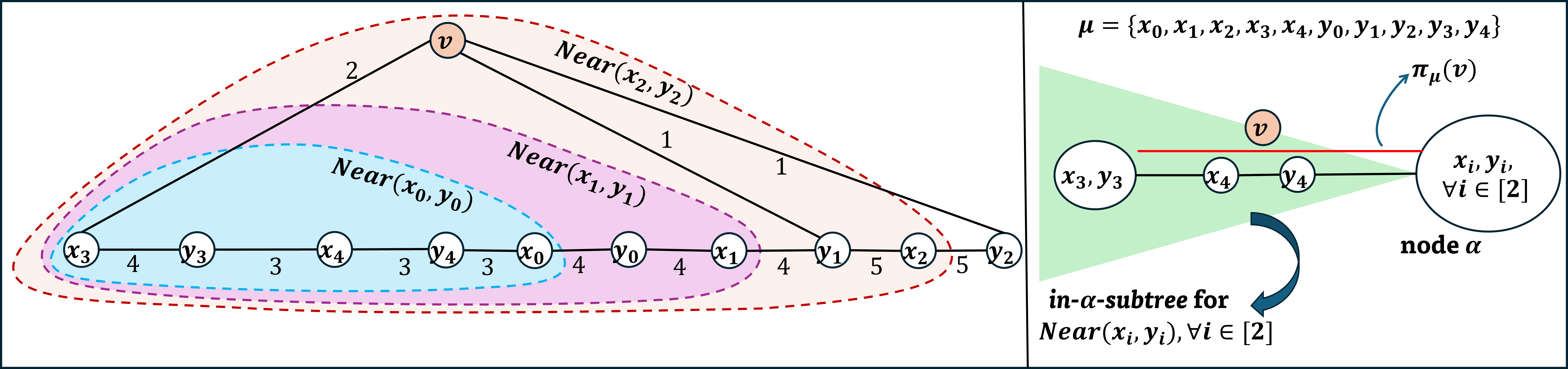}
  \end{center}
  \caption{Example illustrating the hardness in case (III). The left figure is an example graph, and its skeleton $\cH_\mu$ is on the right, where $\mu$ is the terminal set. $\mu$-mincut value is $5$. Projection of nonterminal $v$ is in the configuration (III) for $\Near(x_i,y_i)$, for each $i\in[0,2]$: it has one endpoint in in-$\alpha$-subtree, and the other endpoint is node $\alpha$ in $\cH_\mu$. For $x_1,y_1,x_2,y_2\in \alpha$, $v\notin \Near(x_1,y_1)$ but $v\in \Near(x_2,y_2)$.}
  \label{fig : Chain tool}
\end{figure}
So, the question is, how do we differentiate between the two possibilities? 
To overcome this hurdle, we focus on the family of cuts that defies the intuition from the First Step: 
All pairs $a,b\in \alpha$ whose nearest mincut includes $T_{in}$ but excludes $v$, i.e., $T_{in}\subseteq \Near(a,b)$ but $v\notin \Near(a,b)$.
Denote the family of all such cuts by $N_v^{T_{in}}$.

It is clear that the problem of determining whether $v\in \Near(x,y)$ reduces to the problem of determining whether $\Near(x,y)\in N_v^{T_{in}}$.
However, that membership query seems challenging as the family $N_v^{T_{in}}$ may contain many cuts, and we cannot afford to store them all.
Furthermore, unlike single-source nearest mincuts, the cuts in $N_v^{T_{in}}$ are defined by arbitrary pairs of vertices from $\alpha$, and thus they do not seem to have any obvious relation; 
in principle, they may overlap or cross in an arbitrary manner. 
Contrary to this intuition, the example in \Cref{fig : Chain tool} hints at a much stronger phenomenon. 
The two cuts  $\Near(x_0,y_0), \Near(x_1,y_1)\in N_v^{T_{in}}$ in the figure are actually nested, satisfying $\Near(x_0,y_0)\subseteq \Near(x_1,y_1)$. 
We show that this is not a coincidence; and in fact, every two elements from $N_v^{T_{in}}$ are \textit{related by containment}, and thus $N_v^{T_{in}}$ forms a chain. 
Using this property, it is clear that given the maximal element $\Near(w,z)$ in the chain, $\Near(x,y)\in N_v^{T_{in}}$ if and only if $\Near(x,y)\subseteq \Near(w,z)$. 
Thus, in the preprocessing step, for every vertex $v$ whose projection is stored, we store two chain maximizers, one for each endpoint of its projection, assuming the adjacent edge incident on that endpoint is the in-edge for the (yet unknown) query.

Interestingly, we show that in order to determine if $\Near(x,y)\in N_v^{T_{in}}$, it is sufficient to store only a pair of witness vertices $w,z$ determining the maximal element $\Near(w,z)$, called the \textit{chain maximizer}.
Therefore, we only append $O(1)$ additional information to each stored projection $\pi_\mu(u)$, and the total space remains $O(n)$.
We now present the interesting structural ideas to establish that $N_v^{T_{in}}$ forms a chain. 
Then, to conclude the technical overview, we argue how to use the chain maximizer to efficiently determine if $\Near(x,y)\in N_v^{T_{in}}$.

\paragraph{A chain of nearest mincuts.} 
Our main insight is exposing the following key structural property about any two nearest mincuts, which might be of independent interest. 
\begin{restatable}{lemma}{mainChainLemma} \label{lem: intro: main chain lemma}
    For any pair of cuts $\Near(a,b)$ and $\Near(a',b')$ not related by containment where  $a,b,a',b'\in V$, among the two cuts $\Near(a,b)\setminus \Near(a',b')$ and $\Near(a',b')\setminus \Near(a,b)$, one is an $a,b$-mincut and the other is an $a',b'$-mincut.
\end{restatable} 
The proof is technically involved and relies on a detailed analysis of all possible placements of the vertices $a,b,a',b'$ with respect to the cuts $\Near(a,b)$ and $\Near(a',b')$. 
The argument combines submodularity and posimodularity to establish the desired mincut properties in each possible configuration.

Given the structural property of \Cref{lem: intro: main chain lemma}, the proof that $N_v^{T_{in}}$ forms a chain is an application of the four component lemma (\Cref{lem: four component lemma}) and a careful use of the projection of $\pi_\mu(v)$. 
We give the underlying argument of this proof below.
Assume towards contradiction that there exist two cuts in $N_v^{T_{in}}$ not related by containment, namely $A =\Near(a_1,b_1),B =\Near(a_2,b_2)$.
In addition, one can select a $\mu$-mincut $C$ satisfying $T_{in}\subseteq C$, $\alpha \subseteq \overline{C}$, and $v\in C$. 
Such a cut exists since $\pi_{\mu}(v)\not\subseteq T_{in}$.
It is clear that $T_{in} \subseteq A\cap B\cap C$ and $v\in C\setminus (A\cup B)$. 
Furthermore, since $\alpha \subseteq \overline{C}$ and $a_1,b_1,a_2,b_2\in \alpha$, by \Cref{lem: intro: main chain lemma}, we can show that one of $A\setminus (B\cup C)$ and $B\setminus (A\cup C)$ is an $a_1,b_1$-cut and the other is an $a_2,b_2$-cut.
In addition, $A\cap B\cap C$ is a $\mu$-cut separating $T_{in}$ from $\alpha$, since each contains a vertex from $\mu$.
Applying \Cref{lem: four component lemma}, we obtain that $C\setminus(A\cup B)$ is an empty set.
However, this is a contradiction since $v\in C\setminus(A\cup B)$.

\paragraph{Utilizing Chain Maximizer to determine if $\Near(x,y)$ is in the chain $N_v^{T_{in}}$.} 
Let $w,z\in \alpha$ be such that $\Near(w,z)$ is the maximal element in the chain of $N_v^{T_{in}}$, that is, $\Near(w,z)$ is the chain maximizer.
Recall that $\Near(x,y)\in N_v^{T_{in}}$ if and only if $\Near(x,y)\subseteq \Near(w,z)$.
We show that this containment relation can be determined using only the information known to us, that is $\Near(x,y)\cap \alpha$ and $w,z$. 
Firstly, it is easy to see $\Near(x,y)\not\subseteq \Near(w,z)$ if $z\in \Near(x,y)$, since $z\notin \Near(w,z)$.
Note that $z\in \alpha \subseteq \mu$ (the latter is since $\alpha$ is a child of $\mu$), and we have already classified $\mu\cap \Near(x,y)$, so we can determine whether $z\in \Near(x,y)$ in $O(1)$ time.

For the case $z\notin \Near(x,y)$, a more involved analysis is required. In this case, both $\Near(x,y) \in N_v^{T_{in}}$ and $\Near(x,y)\notin N_v^{T_{in}}$ are possible, and therefore, knowing the position of $w,z$  with respect to $\Near(x,y)$ is no longer sufficient. Interestingly, we show that this situation can be handled if the position of $x$ with respect to $\Near(w,z)$ is known. 
In particular, when $z\not\in \Near(x,y)$, we have \textit{$\Near(x,y)\subseteq \Near(w,z)$ if and only if $x\in \Near(w,z)$.}

Unfortunately, $\Near(w,z)$ is not known to us.
So, how can we determine if $x\in \Near(w,z)$? 
It turns out that our stored information is still sufficient for this purpose. 
To achieve this, we have to again split the analysis into two cases based on whether $w\in \Near(x,y)$ or not.
Specifically, we show that if $w,z\not\in \Near(x,y)$, then $\Near(x,y)\subseteq \Near(w,z)$.
This can be determined in $O(1)$ time since we have already classified $\alpha\cap \Near(x,y)$ and $w,z\in \alpha$.
The only remaining case is now $w\in \Near(x,y)$ and $z\not\in \Near(x,y)$. 
The structural properties of the hierarchy tree $\cT$ turn out to be useful in solving this case. 
In particular, we show that $\LCA(x,y)$ must be a descendant of $\LCA(w,z)$ in $\cT$.
This is sufficient for us to determine whether $x\in \Near(w,z)$ in $O(1)$ time by finding whether $x\in \Near(w,z)\cap \LCA(w,z)$ using just the skeleton $\cH_{\LCA_{\cT}(w,z)}$.
This completes our overview for the generic steps.

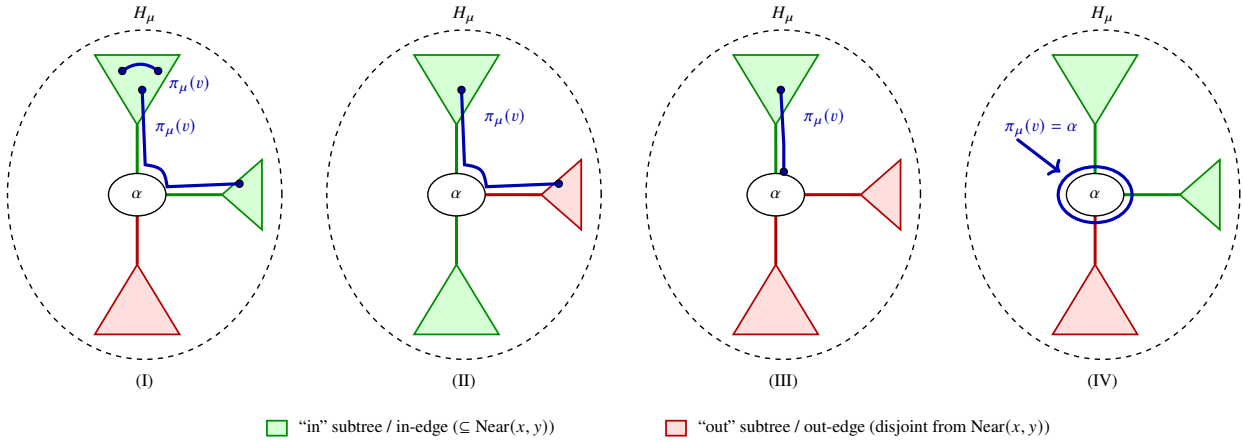
\begin{figure}[h]
        
        \centering
        \scalebox{0.66}{
            \begin{tikzpicture}[
  thick,
  vertex/.style={
    ellipse, draw, fill=white,
    minimum width = 1.15cm, minimum height = 0.85cm,
    inner sep = 1pt, font = \small,
  },
  intri/.style ={draw=green!55!black, fill=green!16, line width=1pt}, 
  outtri/.style={draw=red!70!black,   fill=red!13,   line width=1pt}, 
  inedge/.style ={green!55!black, line width=1.7pt},                  
  outedge/.style={red!70!black,   line width=1.7pt},                  
  proj/.style ={draw=blue!70!black, line width=1.8pt, line cap=round},
  projdot/.style={fill=blue!70!black},
  plabel/.style={font=\small},
]


\begin{scope}[xshift=0cm]
  \draw[dashed] (0.15,0) ellipse (2.75cm and 3.3cm);
  \node[plabel] at (0.15,3.6) {$H_\mu$};

  \draw[inedge]  (0,0) -- (0, 1.4);
  \draw[inedge]  (0,0) -- (1.7,0);
  \draw[outedge] (0,0) -- (0,-1.4);

  \draw[intri]  ( 0, 1.4) -- (-0.85, 2.8) -- (0.85, 2.8) -- cycle;
  \draw[intri]  (1.7, 0 ) -- ( 2.5,-0.7) -- (2.5, 0.7) -- cycle;
  \draw[outtri] ( 0,-1.4) -- (-0.85,-2.8) -- (0.85,-2.8) -- cycle;

  \draw[proj] (0.10,2.10) -- (0.16,0.60)
    .. controls (0.50,0.55) and (0.55,0.50) .. (0.60,0.16)
    -- (2.05,0.22);

  \node[vertex] at (0,0) {$\alpha$};

  \filldraw[projdot] (0.10,2.10) circle (2pt);
  \filldraw[projdot] (2.05,0.22) circle (2pt);

  \node[plabel, text=blue!70!black, anchor=west] at (0.22,1.35) {$\pi_\mu(v)$};

  \draw[proj] (-0.30,2.48) .. controls (-0.05,2.66) and (0.18,2.66) .. (0.42,2.48);
  \filldraw[projdot] (-0.30,2.48) circle (2pt);
  \filldraw[projdot] ( 0.42,2.48) circle (2pt);
  \node[plabel, text=blue!70!black, anchor=west] at (0.50,2.2) {$\pi_\mu(v)$};

  \node[plabel] at (0.15,-3.75) {(I)};
\end{scope}

\begin{scope}[xshift=6.4cm]
  \draw[dashed] (0.15,0) ellipse (2.75cm and 3.3cm);
  \node[plabel] at (0.15,3.6) {$H_\mu$};

  \draw[inedge]  (0,0) -- (0, 1.4);
  \draw[outedge] (0,0) -- (1.7,0);
  \draw[inedge]  (0,0) -- (0,-1.4);

  \draw[intri]  ( 0, 1.4) -- (-0.85, 2.8) -- (0.85, 2.8) -- cycle;
  \draw[outtri] (1.7, 0 ) -- ( 2.5,-0.7) -- (2.5, 0.7) -- cycle;
  \draw[intri]  ( 0,-1.4) -- (-0.85,-2.8) -- (0.85,-2.8) -- cycle;

  \draw[proj] (0.10,2.10) -- (0.16,0.60)
    .. controls (0.50,0.55) and (0.55,0.50) .. (0.60,0.16)
    -- (2.05,0.22);

  \node[vertex] at (0,0) {$\alpha$};

  \filldraw[projdot] (0.10,2.10) circle (2pt);
  \filldraw[projdot] (2.05,0.22) circle (2pt);

  \node[plabel, text=blue!70!black, anchor=west] at (0.42,1.55) {$\pi_\mu(v)$};
  \node[plabel] at (0.15,-3.75) {(II)};
\end{scope}

\begin{scope}[xshift=12.8cm]
  \draw[dashed] (0.15,0) ellipse (2.75cm and 3.3cm);
  \node[plabel] at (0.15,3.6) {$H_\mu$};

  \draw[inedge]  (0,0) -- (0, 1.4);
  \draw[outedge] (0,0) -- (1.7,0);
  \draw[outedge] (0,0) -- (0,-1.4);

  \draw[intri]  ( 0, 1.4) -- (-0.85, 2.8) -- (0.85, 2.8) -- cycle;
  \draw[outtri] (1.7, 0 ) -- ( 2.5,-0.7) -- (2.5, 0.7) -- cycle;
  \draw[outtri] ( 0,-1.4) -- (-0.85,-2.8) -- (0.85,-2.8) -- cycle;

  \draw[proj] (0.10,2.10) -- (0.16,1.00) -- (0.16,0.46);

  \node[vertex] at (0,0) {$\alpha$};

  \filldraw[projdot] (0.10,2.10) circle (2pt);
  \filldraw[projdot] (0.16,0.46) circle (2pt);

  \node[plabel, text=blue!70!black, anchor=west] at (0.42,1.55) {$\pi_\mu(v)$};
  \node[plabel] at (0.15,-3.75) {(III)};
\end{scope}

\begin{scope}[xshift=19.2cm]
  \draw[dashed] (0.15,0) ellipse (2.75cm and 3.3cm);
  \node[plabel] at (0.15,3.6) {$H_\mu$};

  \draw[inedge]  (0,0) -- (0, 1.4);
  \draw[inedge]  (0,0) -- (1.7,0);
  \draw[outedge] (0,0) -- (0,-1.4);

  \draw[intri]  ( 0, 1.4) -- (-0.85, 2.8) -- (0.85, 2.8) -- cycle;
  \draw[intri]  (1.7, 0 ) -- ( 2.5,-0.7) -- (2.5, 0.7) -- cycle;
  \draw[outtri] ( 0,-1.4) -- (-0.85,-2.8) -- (0.85,-2.8) -- cycle;

  \node[vertex] at (0,0) {$\alpha$};

  \draw[proj] (0,0) ellipse (0.74cm and 0.56cm);

  \node[plabel, text=blue!70!black, anchor=west] at (-1.95,1.35) {$\pi_\mu(v)=\alpha$};
  \draw[proj, ->] (-1.55,1.10) -- (-0.70,0.45);

  \node[plabel] at (0.15,-3.75) {(IV)};
\end{scope}

\begin{scope}[yshift=-4.7cm]
  \filldraw[intri]  (2.6,-0.13) rectangle ++(0.4,0.3);
  \node[plabel, anchor=west] at (3.10,0.02)
        {``in'' subtree / in-edge ($\subseteq \mathrm{Near}(x,y)$)};

  \filldraw[outtri] (10.6,-0.13) rectangle ++(0.4,0.3);
  \node[plabel, anchor=west] at (11.10,0.02)
        {``out'' subtree / out-edge (disjoint from $\mathrm{Near}(x,y)$)};

\end{scope}

\end{tikzpicture}
        }
        \caption{Illustration of the four possible configurations for $\pi_\mu(v)$.}
        \label{fig:configurations}
\end{figure}
\paragraph{Organization of the paper.} 
The rest of the paper is devoted to the proof of our results. In \Cref{sec : chain tool}, we present the proofs for the \textit{chain-maximizer} for nearest mincuts. 
Next,  the construction of the Nearest mincut hierarchy is given  in \Cref{sec: data structure}. 
Then, the query algorithm for reporting a nearest mincut is given in \Cref{sec:near-minimum-cut-query}, giving \Cref{thm:main}. The structural results on characterizing skeleton subtrees using projections and chain maximizers are provided in \Cref{sec : characterization of skeleton subtrees}.
Following is the query algorithm for reporting the nearest mincut tree in \Cref{sec:nearest-mincut-tree-construction-query} that establishes \Cref{thm : nearest mincut tree reporting}. Next, we explain our sensitivity oracle from \Cref{thm : mincut sensitivity data structures for the insertion of an edge} in \Cref{sec:sensitivity-oracles-insertion-edge}. The lower bound on the minimal Gomory-Hu tree is given in   \Cref{sec:nearest-gomory-hu-lower-bound}. Finally, we conclude with some interesting future works in \Cref{sec : conclusion}.  
\begin{figure}[H]
  \begin{center}
    \includegraphics[width=\textwidth]{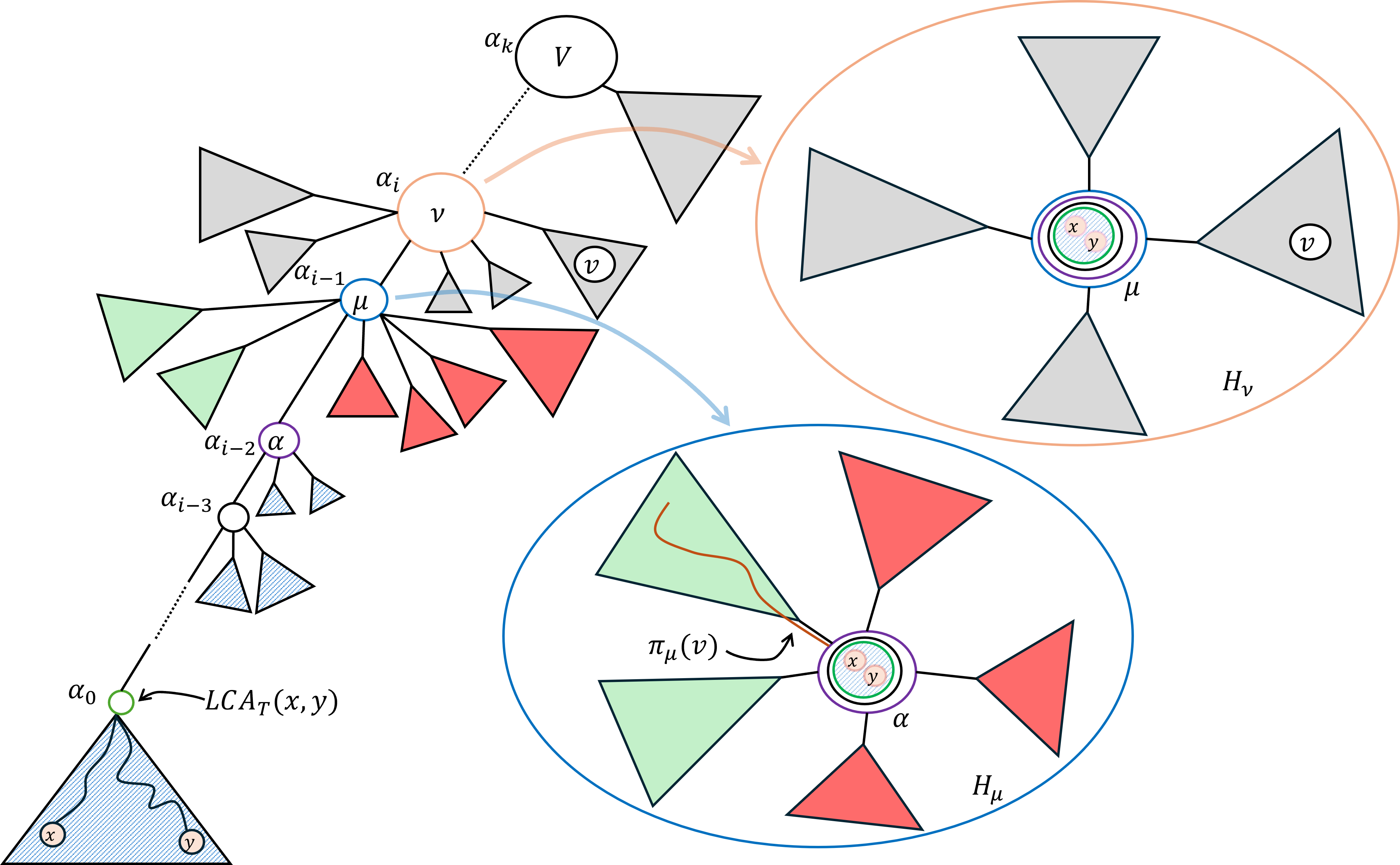}
  \end{center}
  \caption{Illustration of the generic step framework with the nearest mincut hierarchy. The red and blue encircled regions show the $\alpha$-subtrees and $\mu$-subtrees views of Skeletons $\cH_\mu$ and $\cH_\nu$, respectively. The green and red colored triangles are those in-$\alpha$-subtrees and out-$\alpha$-subtrees, respectively. Node $\mu$ and the blue triangles represent the nodes in $\cT$ that are solved in the previous steps of our query algorithm. At the current step, $\Near(x,y)\cap \mu$ is known to us, and we need to solve $\Near(x,y)\cap \nu$. Vertex $v$ is the representative vertex from one of the yet unclassified $\mu$-subtrees (triangles with grey color) such that $\pi_{\mu}(v)$ is in configuration (III). Note that subtrees in Skeletons are different from subtrees rooted at the nodes in $\cT$.}
  \label{fig:genstep}
\end{figure}

\section{Chain Maximizers for a Family of Nearest Mincuts} \label{sec : chain tool}
In this section, we prove the following useful property regarding a family of nearest mincuts between vertices that belong to the endpoints of the projection of any vertex for a given terminal set.
\begin{theorem} [Chain Maximizers of nearest mincuts] \label{thm : chain-maximizer}
    Let $S\subseteq V$ be a terminal set and let $v\in V\setminus S$ be some vertex such that $\alpha,\beta$ are the endpoints of $\pi_S(v)$ with $\alpha\ne \beta$ and their adjacent edges are $g_\alpha,g_\beta$ respectively.
    Let $T_\alpha$ be the $\alpha$-subtree of $\cH_S$ rooted at the opposite end of $g_{\alpha}$ from $\alpha$;
    and define $T_{\beta}$ analogously.
    Define the set 
    \begin{equation*}
    N_v^{T_{\alpha}}  
    = \left\{
        \Near(a,b) 
        \middle| 
    \begin{array}{l}
    a,b\in \alpha \\
    v\notin \Near(a,b) \\
    T_\alpha \subseteq \Near(a,b)
    \end{array}
    \right\}    
    ,
    \end{equation*} 
    and the set $N_v^{T_{\beta}}$ is defined symmetrically. 
    Then, every two elements $A,B\in N_v^{T_{\alpha}}$ (or $A,B\in N_v^{T_{\beta}}$) are related by containment, i.e., either $A\subseteq B$ or $B\subseteq A$.
    Furthermore, there exists a unique inclusion-wise maximal element in $N_v^{T_{\alpha}}$ and in $N_v^{T_{\beta}}$.
\end{theorem}

To prove \Cref{thm : chain-maximizer}, we first prove the following two lemmas, which are of independent interest and might be useful in other contexts as well.
The first lemma proves a structural result on crossing minimum cuts within a certain special setting, which is illustrated in \Cref{fig:crossing-mincuts}. 
We then use this lemma to prove the main technical result for the chain, \Cref{lem: intro: main chain lemma}, as presented in the technical overview section.
\begin{lemma} [All-Pairs Crossing Mincuts Lemma] \label{lem: crossing mincuts with special assignment of vertex pairs}
    Let $a,b,a',b'\in V$ and suppose that a $(a,b)$-mincut $C\subseteq V$ crosses with a $(a',b')$-mincut $C'$ such that $a\in C\cap C'$, $b\in \overline{C\cup C'}$, $a'\in C'\setminus C$, and $b'\in C\setminus C'$. Then,
    \begin{enumerate}
        \item $c(C)=c(C')$, 
        \item $C\cap C'$ and $C\cup C'$ are $(a,b)$-mincuts, and
        \item $C\setminus C'$ is a $(a',b')$-mincut and $C'\setminus C$ is a $(b',a')$-mincut.
    \end{enumerate}
\end{lemma}
\begin{proof}
    Since $a\in C\cap C'$ and $b\notin C\cup C'$, the cuts $C\cap C'$ and $C\cup C'$ are $(a,b)$-cuts. 
    Similarly, $C\setminus C'$ is a $(a',b')$-cut and $C'\setminus C$ is a $(b',a')$-cut. 
    It follows that $c(C\cup C')\ge c(C)$, $c(C\cap C')\ge c(C)$, $c(C\setminus C')\ge c(C')$, and $c(C'\setminus C)\ge c(C')$.
    Therefore, by the submodularity of the cut function (\Cref{lem:sub-posi-general}) applied to $C,C'$, we get that $c(C')\ge c(C)$. 
    Similarly, by applying posimodularity (\Cref{lem:sub-posi-general}) on cuts $C,C'$, we get that $c(C)\ge c(C')$. 
    It immediately follows that $c(C)=c(C')$, and furthermore, $c(C\cup C')=c(C\cap C')=c(C\setminus C')=c(C'\setminus C)=c(C)$.
    This concludes the proof. 
\end{proof}

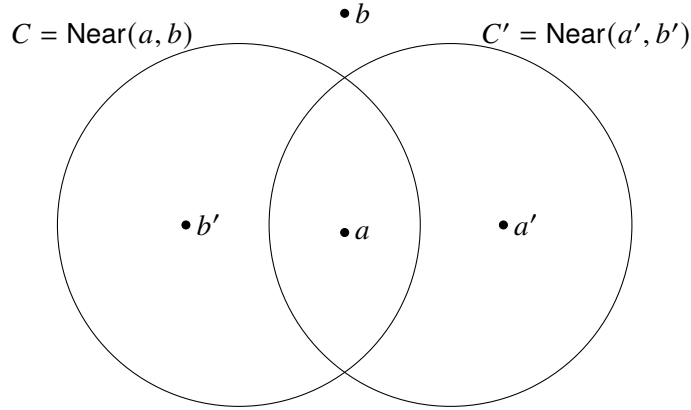
\begin{figure}[h]
        \label{fig:all-pairs-crossing-mincuts}
        \centering
        \begin{tikzpicture}

  \draw (-1.4, 0) circle (2.4cm);

  \draw (1.4, 0) circle (2.4cm);

  \node[font=\itshape] at (-3.2, 2.5) {$C=\Near(a,b)$};

  \node[font=\itshape] at (3.2, 2.5) {$C'=\Near(a',b')$};

  \filldraw (0, 2.8) circle (1.5pt);
  \node[right, font=\itshape] at (0, 2.8) {$b$};

  \filldraw (-2.1, 0) circle (1.5pt);
  \node[right, font=\itshape] at (-2.1, 0) {$b'$};

  \filldraw (0, -0.1) circle (1.5pt);
  \node[right, font=\itshape] at (0, -0.1) {$a$};

  \filldraw (2.1, 0) circle (1.5pt);
  \node[right, font=\itshape] at (2.1, 0) {$a'$};

\end{tikzpicture}
        \caption{Illustration of the setting of \Cref{lem: crossing mincuts with special assignment of vertex pairs}.}
        \label{fig:crossing-mincuts}
\end{figure}

\mainChainLemma*
\begin{proof}
    The proof follows by an exhaustive analysis of all possible configurations of the vertices $a,b,a',b'$ with respect to the cuts $\Near(a,b)$ and $\Near(a',b')$.
    Denote $C\coloneqq\Near(a,b)$ and $C'\coloneqq\Near(a',b')$.
    Suppose  $a\in C\setminus C'$ and $a'\in C'\setminus C$. 
    Then, using posimodularity (\Cref{lem:sub-posi-general}), it is immediate that $C\setminus C'$ is an $(a,b)$-mincut and $C'\setminus C$ is an $(a',b')$-mincut.

    Notice that whenever $C\cap C'$ is an $(a,b)$-cut and $C\cup C'$ is an $(a',b')$-cut (or vice versa), then by submodularity (\Cref{lem:sub-posi-general}), $C\cap C'$ is an $(a,b)$-mincut and $C\cup C'$ is an $(a',b')$-mincut.
    Then, by the minimality of $C$ we have that $C\subseteq C\cap C'$, which contradicts the assumption that $C$ and $C'$ are not related by containment.
    Therefore, we can disregard all these cases.


    Suppose both $a,a'\in C\cap C'$. Notice that if at least one of $b,b'$ is outside $C\cup C'$, then $C\cap C'$ is an $(a,b)$-cut and $C\cup C'$ is an $(a',b')$-cut (or vice versa), and we can disregard this case as explained above.
    Hence, $b'\in C\setminus C'$ and $b\in C'\setminus C$, and we again obtain that $C\setminus C'$ is an $a',b'$-mincut and $C'\setminus C$ is an $a,b$-mincut by posimodularity (\Cref{lem:sub-posi-general}). 

    We now assume that exactly one of $a,a'$, say $a$, is in $C\cap C'$; otherwise, the proof is symmetric.
    Notice that if $b'\in \overline{C\cup C'}$, then $C\cap C'$ is an $(a,b)$-cut and $C\cup C'$ is an $(a',b')$-cut, and we again disregard this case as explained above.
    Therefore, $b'\in C\setminus C'$.

    To find the location of $b$, we have to leverage \Cref{lem: crossing mincuts with special assignment of vertex pairs}.
    If $b \notin C\cup C'$, then we obtain the configuration illustrated in \Cref{fig:crossing-mincuts} since $a'\in C'\setminus C$. Thus, by \Cref{lem: crossing mincuts with special assignment of vertex pairs}(2), $C\cap C'$ is a $(a,b)$-mincut.
    Again, this implies that $C\subseteq C\cap C'$, a contradiction to the assumption that $C$ and $C'$ are not related by containment.
    Therefore, we only remain with the case where $b\in C'\setminus C$.
    Then, by posimodularity (\Cref{lem:sub-posi-general}), $C\setminus C'$ is a $a',b'$-mincut and $C'\setminus C$ is a $a,b$-mincut.
    This concludes the proof of the lemma.
\end{proof}
We are now ready to prove \Cref{thm : chain-maximizer}.
\begin{proof}[Proof of \Cref{thm : chain-maximizer}]
    We prove the lemma for $N_v^{T_{\alpha}}$; the proof for $N_v^{T_{\beta}}$ is symmetric.
    Let $A,B\in N_v^{T_\alpha}$ be two cuts such that $A=\Near(a,b),B=\Near(a',b')$, and assume towards contradiction that $A$ and $B$ are not related by containment.
    In addition, let $C$ be the $S$-mincut defined by the $g_\alpha\in \pi_S(v)$ such that $v\in C$ and $\alpha\not\subseteq C$.
    Notice that by the definition of $C$ we have $T_{\alpha}\subseteq C$ and $a,b,a',b'\notin C$.
    Such a cut exists by \Cref{thm: proj vertex}.

    
    By \Cref{lem: intro: main chain lemma}, we have that among the two cuts $A\setminus B$ and $B\setminus A$, one is an $a,b$-cut and the other is an $a',b'$-cut.
    Furthermore, since $a,b\in \alpha$ and $\alpha\subseteq \overline{C}$, it follows that among the two cuts $A\setminus (B\cup C)$ and $B\setminus (A\cup C)$, one is an $a,b$-cut and the other is an $a',b'$-cut.
    In addition, since $T_{\alpha} \subseteq A\cap B\cap C$ and $\alpha \subseteq \overline{C}$, it follows that $A\cap B\cap C$ is an $S$-cut.
    Applying \Cref{lem: four component lemma} on the cuts $A,B,C$, we get that $c(C\setminus(A\cup B))=0$.
    However, $v\in C\setminus(A\cup B)$ and $G$ is connected, a contradiction.
\end{proof}

\section{Data Structure: Nearest Mincut Hierarchy} \label{sec: data structure}

In this section, we prove \Cref{lem: nearest mincut hierarchy space}, giving the space guarantees of the nearest mincut hierarchy.
The details of the query algorithm for recovering $\Near(x,y)$ given any two vertices $x,y\in V$ is deferred to \Cref{sec:near-minimum-cut-query}.
As detailed in the introduction, the key components of our data structure are the following: 
\begin{itemize}
    \item Skeleton Hierarchy from \Cref{sec: skeleton hierarchy},
    \item Projections of vertices from \Cref{thm: proj vertex}, and
    \item Chain-maximizers for nearest mincuts from \Cref{thm : chain-maximizer}.
\end{itemize}

Our construction begins with the Skeleton Hierarchy tree of \Cref{sec: skeleton hierarchy}. 
Recall that the Skeleton hierarchy consists of the two components: the hierarchy tree $\cT$ of \cite{KatzKKP04} and Skeletons $\cH_\mu$ from \Cref{thm:skeleton} that are augmented to each non-leaf node $\mu$ in $\cT$. Recall that for any node $\mu\in \cT$,  the corresponding vertex set is also defined by $\mu$. The following fact about $\cT$ is essential for our query algorithms.
\begin{fact}\label{claim:mincut-value-increases}
Hierarchy tree $\cT$ satisfies the property:  
    Let $\alpha$ be an ancestor of a node $\beta$ in $\cT$, then the value of  $\alpha$-mincut is strictly smaller than the value of  $\beta$-mincut.
    
\end{fact}
The nearest mincut hierarchy is obtained by augmenting this Skeleton hierarchy with the vertex projections and chain maximizers as follows.

\paragraph{Vertex Projections.} 
We now describe the set of vertex projections that are stored in our data structure.
Let $\mu$ be any node in $\cT$, and let $\nu$ be its parent in $\cT$. 
Observe that $\mu$ is a node in the Skeleton $\cH_\nu$.
As described in \Cref{sec:technical-overview}, our query algorithm's strategy is to classify for every $\mu$-subtree in $\cH_\nu$ whether it is in $\Near(x,y)$ or not.
Therefore, for every $\mu$-subtree $T_\mu$ in $\cH_\nu$, select an arbitrary vertex $v\in T_\mu$ and call it the \textit{representative vertex} of $T_\mu$.
Then, store the projection of the representative vertex $v$ in the Skeleton as follows.
Whenever $\pi_\mu(v)$ is not a single nonempty node, simply store its two endpoints (or its single endpoint if it is an empty node). 
Otherwise, recursively move down into the node $\pi_\mu(v)$ and store its projection there.
Repeat this process until it reaches a nonempty node or a leaf in the hierarchy tree $\cT$.
Notice that the nodes which are reached in this process form a path in the hierarchy tree $\cT$ from $\nu$ to some node $\rho$.
Furthermore, for every node $\rho$ along this path, one can recover the projection $\pi_\rho(v)$ of $v$, since it is exactly the child in which the projection of $v$ is ultimately contained.
This process is formalized in \Cref{alg: projection} below, which is invoked by assigning $u\gets v$ and $\beta\gets \mu$.



\begin{algorithm} 
\caption{\textsc{Store-Projection}$(u,\beta)$}
\begin{algorithmic}[1]
\If{$\pi_\beta(u)$ is not a single nonempty node or $\pi_\beta(u)$ is a leaf in $\cT$}
\Statex Store the endpoints of $\pi_\beta(u)$;
\Else~\textsc{Store-Projection}$(u,\alpha)$, where $\pi_\beta(u)=\alpha$;
\EndIf
\label{alg: projection}
\end{algorithmic}
\end{algorithm}

\paragraph{Chain Maximizers.} 
The final component of our data structure are the chain-maximizers for nearest mincuts from \Cref{thm : chain-maximizer}.
Again consider the representative vertex $v$ for the $\mu$-subtree defined above.
Let $\alpha,\beta$ be the endpoints of $\pi_\gamma(v)$ such that $\alpha\ne \beta$, where $\gamma$ is the first node in $\cT$ in which the projection of $v$ is a path.
If $\alpha=\beta$ ($v$ is projected in an empty node/leaf), do not store maximizers for $v$.
From \Cref{thm : chain-maximizer} there exists a unique maximal element in each of the families $N_v^{T_\alpha}, N_v^{T_\beta}$, where $T_\alpha,T_\beta$ are defined as in \Cref{thm : chain-maximizer}.
Then, store the maximal elements of $N_v^{T_\alpha}, N_v^{T_\beta}$ alongside the projection of $v$.
Specifically, choose a pair of vertices $a,b$ (resp., $a',b'$) for which $\Near(a,b)$ (resp., $\Near(a',b')$) is the maximal element in $N_v^{T_\alpha}$ (resp., $N_v^{T_\beta}$), and store the pair $(a,b)$ (resp., $(a',b')$) with the projection of $v$ in our data structure.
Finally, store two additional flag bits $M_v^\alpha$ and $M_v^\beta$ with $v$, and set $M_v=1$ if $N_v^{T_\alpha}$ is an empty set; otherwise, $M_v=0$. 
Set $M_v^\beta$ similarly for $N_v^{T_\beta}$.


Finally, our complete data structure is obtained by augmenting the hierarchy tree $\cT$ and each of the stored skeletons $\cH_\mu$ for every internal node $\mu$ of $\cT$ with the following tree data structures. This augmentation is required only to answer our queries faster.  

\begin{fact}[\cite{HT84, BenderF04}] \label{fact: tree data structures}
    For any rooted tree $T$ on $n$ vertices, there is an $O(n)$ space data structure that
    \begin{itemize}
        \item given any vertices $u,v$ in $T$, reports in $O(1)$ time the lowest common ancestor ($\LCA$) of $u,v$  in $T$.
        \item given any vertex $u$ in $T$ and depth $d \leq \mathrm{depth}(v)$, reports in $O(1)$ time the ancestor of $v$ at depth $d$ in $T$. Here, \emph{depth} of a vertex in $T$ is its distance from the root (so the root has depth $0$).
    \end{itemize}
\end{fact}

We are now ready to prove \Cref{lem: nearest mincut hierarchy space}.
\begin{proof}[Proof of \Cref{lem: nearest mincut hierarchy space}]
We now show that the space occupied by our data structure is $O(n)$. 
Recall that Skeleton hierarchy occupies $O(n)$ space: $O(n)$ space for storing $\cT$ and all the stored Skeletons in $\cT$  take $O(n)$ space in total. 
Thus, augmenting Skeleton hierarchy with the tree data structures of \Cref{fact: tree data structures} also takes $O(n)$ space. 

It remains to argue that the augmentation of projections and chain-maximizers takes overall $O(n)$ space.
Notice that for each stored projection, we store exactly two chain-maximizers, and hence the space taken by chain maximizers and projections is $O(k)$ as each one is encoded in $O(1)$ space, where $k$ is the number of stored projections.
Hence, it suffices to show that $k=O(n)$.
For every node $\mu$ in $\cT$, we store $O(\deg_{\cH_\nu}(\mu))$ projections, where $\nu$ is the parent of $\mu$ in $\cT$. 
Notice that $\sum_{\mu \text{ child of } \nu} O(\deg_{\cH_\nu}(\mu)) = O(|\cH_\nu|)$.
Therefore, the number of projections stored in the Skeleton $\cH_\nu$ is $O(|\cH_\nu|)$, and summing over all the nodes in $\cT$, we get that the total number of projections stored is $O(\sum_{\mu\in \cT} |\cH_\mu|)=O(n)$.
\end{proof}

\section{A Characterization of Skeleton Subtrees}
\label{sec : characterization of skeleton subtrees}
In this section, we establish our key structural insights to characterize a subtree in a Skeleton from Nearest mincut hierarchy with respect to $\Near(x,y)$. 
The characterization is useful both in our query algorithm for efficiently reporting $\Near(x,y)$ for a given pair of vertices $x,y$ in \Cref{sec:near-minimum-cut-query} and in designing sensitivity oracles in \Cref{sec:sensitivity-oracles-insertion-edge}. This result is summarized in the following theorem.
\begin{theorem}[Subtree Characterization]\label{thm: characterization of skeleton subtrees}
    Let $x,y\in V$ be two vertices and $\mu$ be a node in the hierarchy tree $\cT$ of nearest mincut hierarchy such that $x,y\in \mu$ and $\pi_\mu(x)=\alpha$.
    In addition, let $\nu$ be any ancestor of $\mu$ in $\cT$.
    Suppose in-$\alpha$-subtrees (and out-$\alpha$-subtrees) are known with respect to $\Near(x,y)$. 
    Then, for any vertex $v\in \nu\setminus \pi_{\nu}(x)$ satisfying $\pi_\mu(v)\ne \alpha$, the following characterizations hold: 
    \begin{enumerate}[label=\Roman*.]
        \item Suppose both endpoints of $\pi_\mu(v)$ belong to in-$\alpha$-subtrees. 
        Then, $v\in \Near(x,y)$.
        \item Suppose at least one endpoint of $\pi_\mu(v)$ is in an out-$\alpha$-subtree. 
        Then, $v\notin \Near(x,y)$.
        \item Suppose one endpoint of $\pi_\mu(v)$ is $\alpha$ and the other endpoint is in an in-$\alpha$-subtree. 
        Then, given the chain-maximizer $w,z$ of $v$ for the endpoint $\alpha$ of $\pi_\mu(v)$,  
        \begin{enumerate}
            \item if $z\in \Near(x,y)$, then $v\in \Near(x,y)$ and 
            \item if $z\not\in \Near(x,y)$, then $v\in \Near(x,y)$ if and only if $x\notin \Near(w,z)$.
            \item  If $w,z\notin \Near(x,y)$, then $v\notin \Near(x,y)$. 
    \end{enumerate}
    \end{enumerate}
\end{theorem}
\begin{proof}[Proof of \Cref{thm: characterization of skeleton subtrees}]       
We start from case (I).
Its proof is divided into two subcases, when $\pi_{\mu}(v)$ is fully contained in a single in-$\alpha$-subtree of $\cH_{\mu}$, and when it traverses two in-edges of $\alpha$ in $\cH_{\mu}$.
The first case was already handled in \Cref{lem: fully contained in case}, and we give the second case now.
\begin{lemma}\label{lem: two in edges}
    Suppose $\pi_{\mu}(v)$ contains two in-edges of $\alpha$. Then, $v\in \Near(x,y)$.
\end{lemma}
\begin{proof}
    Let $g_1,g_2$ be the two in-edges of $\alpha$ that belong to $\pi_\mu(v)$.
    Let $T_1,T_2$ respectively be the two $\alpha$-subtrees corresponding to $g_1,g_2$, and recall that $T_1,T_2\subseteq \Near(x,y)$.
    Since $\pi_{\mu}(v)$ contains $g_1$, we can use \Cref{thm: proj vertex} to choose an $\mu$-mincut $A_1$ defined by $g_1$ such that $T_1\subseteq A_1$, $v\in A_1$, and $\alpha\subseteq \overline{A_1}$.
    Similarly, we can choose a cut $A_2$ defined by $g_2$ such that $T_2\subseteq A_2$, $v\in A_2$, and $\alpha\subseteq\overline{A_2}$.
    It follows that $\alpha\subseteq \overline{A_1\cup A_2}$;
    and in turn $x\in \overline{A_1\cup A_2}$. Observe that $y\notin T_1\cup T_2$ since they are in-$\alpha$-subtrees. Therefore, $y\in \overline{A_1\cup A_2}$. 
    We now apply the four component lemma (\Cref{lem: four component lemma}) on the three cuts $A_1,A_2,$ and $B=\overline{\Near(x,y)}$.
    Notice that $A_1\setminus (A_2\cup B)$ and $A_2\setminus (A_1 \cup B)$ are $\mu$-cuts since they separate $\alpha$ from $T_1,T_2$ respectively.
    In addition, $B\setminus (A_1\cup A_2)$ is an $x,y$-cut since $x\notin B$ and $y\in B$.
    Therefore, by applying \Cref{lem: four component lemma} on these three cuts, we get that $c(A_1\cap A_2\cap B)=0$.
    Since $v\in A_1,A_2$ and $A_1\cap A_2\cap B=\emptyset$, we conclude that $v\notin B$; otherwise, $G$ becomes disconnected. Hence $v\in \overline{B}=\Near(x,y)$.
\end{proof}
\begin{lemma}[Case (II)]\label{lem: projection touches out edge}
    Suppose $\pi_{\mu}(v)$ contains an out-edge of $\alpha$.
    Then, $v\not\in \Near(x,y)$.
\end{lemma}
\begin{proof}
    Let $T'$ be an out-$\alpha$-subtree that belongs to $\pi_{\mu}(v)$, and let $g$ be the out-edge connecting $T'$ to $\alpha$.
    Consider a $\mu$-mincut $A$ defined by $g$ such that $v\notin A$, $T'\subseteq \overline{A}$, and $\alpha\subseteq A$.
    Such a cut exists by \Cref{thm: proj vertex} as $\pi_{\mu}(v)$ contains $g$.
    Observe that $x\in A\cap \Near(x,y)$, $y\notin A\cap \Near(x,y)$ since $y\not\in \Near(x,y)$. Also, $A\cup \Near(x,y)$ separates $\alpha$ from $T'$.
    It follows that $A\cap \Near(x,y)$ is an $(x,y)$-cut and $A\cup \Near(x,y)$ is an $\mu$-cut.
    Hence, $A\cap \Near(x,y)$ is an $(x,y)$-mincut by applying submodularity on the cuts $A$ and $\Near(x,y)$.
    By the minimality of $\Near(x,y)$, we get that $\Near(x,y)\subseteq A\cap \Near(x,y)$, and hence $v\notin \Near(x,y)$.
\end{proof}

The final case (III) is when $\pi_\mu(v)$ contains exactly one in-edge of $\alpha$. Let $T'$ be the subtree connected with this in-edge. 
As discussed in \Cref{sec:technical-overview}, this case is more intricate, and we need to utilize the chain maximizers  to characterize $v$ as being in or out with respect to $\Near(x,y)$. 
Let $w,z$ be the stored chain maximizer from  \Cref{thm : chain-maximizer} with vertex $v$ for in-$\alpha$-subtree $T'$ and endpoint $\alpha$ of $\pi_\mu(v)$. 
We split case (III) into three subcases, depending on the position of $w,z$ with respect to $\Near(x,y)$, and the position of $x$ with respect to $\Near(w,z)$.
Observe that subcase 3 is superfluous as it is superseded by subcase 2.
However, we retain subcase 3 to bypass the need to solve $\Near(w,z)$ in our query algorithms given in the following sections.  
\begin{lemma} \label{lem: case ii configuration III z is in or wz out}
    Suppose $\pi_\mu(v)$ contains exactly one in-edge of $\alpha$.
    Then,
    \begin{enumerate}
        \item if $z\in \Near(x,y)$, then $v\in \Near(x,y)$,
        \item if $z\not\in \Near(x,y)$, then $v\in \Near(x,y)$ if and only if $x\notin \Near(w,z)$, and
        \item if $w,z\notin \Near(x,y)$, then $v\notin \Near(x,y)$.
    \end{enumerate}
\end{lemma}
\begin{proof}
    Let $g$ be the in-edge in $\pi_{\mu}(v)$, and let $T'$ be the in-$\alpha$-subtree connected to $g$.
    For brevity, denote $W=\Near(w,z)$ and $X=\Near(x,y)$.
    Note that $T'\subseteq W$ by the definition of the family given by \Cref{thm : chain-maximizer}.
    In addition, since $g$ is an in-edge, we have that $T'\subseteq X$.
    Recall that $v\in \Near(x,y)$ if and only if $\Near(x,y)\not\subseteq \Near(w,z)$ by \Cref{thm : chain-maximizer}, and hence we focus on showing this condition.
    We fork on three cases based on the position of $w,z$ w.r.t. $X$.
    
    (1) Suppose $z\in X$. 
    Therefore, $X\not\subset W$ since $z\notin W$. So, by the choice of $W$, $v\in X$.
    
    (2)
    This proof is further split into two cases, depending on if $x\in W$ or not.
    If $x\not\in W$, then $X$ cannot be a subset of $W$ and $v\in X$.
    Now, suppose $x \in W$.
    Then $X \cap W$ is an $x,y$-cut, and $X \cup W$ is a $w,z$-cut since $z\not\in X$ by the assumption of this case.
    Hence, by submodularity (\Cref{lem:sub-posi-general}), $X \cap W$ is an $(x,y)$-mincut.
    It follows $\Near(x,y) \subseteq \Near(w,z)$ by the minimality of $\Near(x,y)$. Hence, $v\not\in \Near(x,y)$.

    (3) 
    Suppose $w,z\notin \Near(x,y)$.
    Observe that if $x\in \Near(w,z)$ then $v\notin \Near(x,y)$ by subcase 2.
    Therefore, assume towards contradiction that $x\not\in \Near(w,z)$.
    Then, $X\setminus W$ is an $(x,y)$-cut and $W\setminus X$ is a $(w,z)$-cut, where we used that $w,z\notin X$ and $x\notin W$.
    By posimodularity, $X\setminus W$ is an $(x,y)$-mincut and $W\setminus X$ is a $(w,z)$-mincut, and hence $X\cap W=\emptyset$ by the minimality of $\Near(x,y)$ and $\Near(w,z)$.
    However, this contradicts the fact that $T'\subseteq X\cap W$ since $T'\subseteq \Near(x,y)$ and $T'\subseteq \Near(w,z)$.
\end{proof} 
This completes our proof of \Cref{thm: characterization of skeleton subtrees}.
\end{proof}

\section{Algorithm for Reporting the Nearest Mincut for a Pair of Vertices}
\label{sec:near-minimum-cut-query}
 
In this section, we establish \Cref{lem: nearest mincut hierarchy query} by presenting our efficient algorithm that uses the nearest mincut hierarchy of $G$ from \Cref{lem: nearest mincut hierarchy space} to report $\Near(x,y)$, for any given ordered pair of vertices $(x,y)\in V$. 

\subsection{Query Algorithm}
The algorithm works by recursively ascending hierarchy tree $\cT$ in the nearest mincut hierarchy, starting from $\LCA_{\cT}(x,y)$ to the root. 
Let us denote the path in $\cT$ from $\LCA_{\cT}(x,y)$ to $V$ by $\LCA_{\cT}(x,y)=\alpha_1,\ldots,\alpha_k=V$. 
The algorithm maintains at every iteration $\Near(x,y)\cap\alpha_i$. By \Cref{fact:lca-mincut-value}, $\LCA_{\cT}(x,y)$-mincut value is the same as $x,y$-mincut value, and by construction of $\cT$, it is the first level where $x,y$ are separated.
In the first iteration, we determine $\Near(x,y)\cap\alpha_1$ using the skeleton $\cH_{\alpha_1}$, as discussed in \Cref{sec : base case}.
Then, in any $i$-th iteration, we determine $\Near(x,y)\cap\alpha_{i+1}$ using the skeletons $\set{\cH_{\alpha_j}}_{j\le i}$, projection and chain-maximizers of representative vertices from $\alpha_i$, and $\Near(x,y)\cap\alpha_{i}$.
The pseudocode is given in \Cref{alg:query}. 

\begin{algorithm}[H]
\caption{Reporting the nearest mincut $\Near(x,y)$ of graph $G=(V,E)$}
\label{alg:query}
\begin{algorithmic}[1]

\Require Nearest mincut hierarchy of $G$ and an ordered query pair $(x,y)\in V\times V$
\State Initialize $\Near(x,y)\gets \{x\}$;
\State $\mu' \gets \LCA_{\cT}(x,y)$;

\State Update $\Near(x,y)\gets \Near(x,y)\cap{\mu'}$ using the Skeleton $\cH_{\mu'}$ (\Cref{thm:skeleton}); 

\While{$\mu'$ is not the root of $\cT$}
    \State $\nu \gets \text{parent of } \mu \text{ in }{\cT}$;
    \For{each $\mu'$-subtree $A$ of Skeleton $\cH_\nu$}
        \State $v\gets $ representative vertex for $A$;
        \State $\rho \gets$ the node in $\cT$ where $\pi_\rho(v)$ is stored.
        \State $\mu\gets \LCA_{\cT}(\rho,\LCA_{\cT}(x,y))$
        \State $(w,z)\gets \text{Chain Maximizer for $v$ with endpoint $\pi_\mu(x)$}$ (\Cref{thm : chain-maximizer});
        \State Classify $A$ using \Cref{thm: characterization of skeleton subtrees}  and $\pi_\mu(v),(w,z)$.
    \EndFor
    \State Update $\Near(x,y)\gets \Near(x,y)\cap\nu$;
    \State $\mu' \gets \nu$;
\EndWhile

\State \Return $\Near(x,y)$

\end{algorithmic}
\end{algorithm}

\subsection{Correctness}
In this section we prove the correctness of \Cref{alg:query}.
Our main tool is the Skeleton subtree characterization of \Cref{thm: characterization of skeleton subtrees}. 
Throughout this section, fix an $\alpha_{i}$-subtree $T$ of $\cH_{\alpha_{i+1}}$ and let $v\in T$ be the representative of $T$.
Recall that by \Cref{lem: subtree classification}, every $\alpha_{i}$-subtree of $\cH_{\alpha_{i+1}}$ is either fully contained in $\Near(x,y)$ or fully disjoint from $\Near(x,y)$.
Therefore, it is sufficient to determine whether $v\in \Near(x,y)$ or not, in order to determine whether $T\subseteq \Near(x,y)$ or $T\cap \Near(x,y)=\emptyset$.  
Our main procedure, composing the main loop of \Cref{alg:query}, is given in the following lemma.
\begin{lemma}
    \label{lem: near from one level to the next}
    There exists an algorithm that, given $\Near(x,y)\cap \alpha_{i}$, can report $\Near(x,y)\cap \alpha_{i+1}$ in time $O(|\alpha_{i+1}\setminus \alpha_i|)$ using nearest mincut hierarchy of $G$.
\end{lemma}
\begin{proof}
    Our algorithm works as follows. 
    It iterates over all $\alpha_{i}$-subtrees of $\cH_{\alpha_{i+1}}$ and classifies each subtree with respect to $\Near(x,y)$ with its representative vertex.
    Then, for every subtree $T$ that is classified as being contained in $\Near(x,y)$, it marks all vertices in $T$ as being in $\Near(x,y)$.
    Note that marking all vertices in $T$ takes time $O(|T|)$, and hence the total time for marking all vertices in $\alpha_{i+1}$ that are in $\Near(x,y)$ is $O(|\alpha_{i+1}\setminus \alpha_i|)$.
    In addition, as we show next, it takes $O(1)$ time to classify each $\alpha_{i}$-subtree, of which there are at most $O(|\alpha_{i+1}\setminus \alpha_i|)$, and hence the total time complexity of the algorithm is $O(|\alpha_{i+1}\setminus \alpha_i|)$.

    For the remainder of the proof, fix an $\alpha_{i}$-subtree $T$ of $\cH_{\alpha_{i+1}}$, and we show how to classify $T$ with respect to $\Near(x,y)$ in $O(1)$ time.
    Let $v$ be the representative terminal of $T$. To determine whether $v\in \Near(x,y)$, we need the two components stored with $v$ -- projection and chain maximizer. Recall that, by \Cref{alg: projection}, these two components for $v$ are stored at a descendant $\gamma$ of $\alpha_{i+1}$ in $\cT$, which is not necessarily the child $\alpha_i$ of $\alpha_{i+1}$. 
    Let $j$ be the first level such that $\pi_{\alpha_j}(v)\ne \pi_{\alpha_{j}}(x)$ or $j=1$ if for all $p\ge 1$ we have $\pi_{\alpha_p}(v)=\pi_{\alpha_{p}}(x)$. Note that when $\gamma$ is not on the path $\alpha_1,\ldots,\alpha_k$, then first determine $j$ in constant time by performing an $\LCA$ query on $\gamma$ and $\LCA_{\cT}(x,y)$ in $\cT$ of the nearest mincut hierarchy; that is $\LCA_\cT(\gamma,\alpha_1)$. 
    Then, we can determine $\pi_{\alpha_j}(v)$ using the level ancestor data structure in $O(1)$ time, by finding the child of $\alpha_j$ that is the ancestor of $\gamma$ in $\cT$. In the other cases, we have $\gamma$ is one of $\alpha_j$, where $j\in[1,k]$.   

    If $j=1$ and $\pi_{\alpha_1}(v)=\alpha_1(x)$, then $v\in \Near(x,y)$, as discussed in the easy base case in \Cref{sec : base case}.
    For the remainder of the proof, we assume that $j\ge 1$ and $\pi_{\alpha_j}(v)\ne \pi_{\alpha_{j}}(x)$. 
    This handles the configuration (IV) discussed in \Cref{sec: overview generic step}. 
    Hence, we only need to handle configurations (I-III) from  \Cref{thm: characterization of skeleton subtrees}.
    
    Throughout, denote $\alpha=\pi_{\alpha_j}(x)$.
    By our assumption, the algorithm has already determined for each edge $g$ adjacent to $\alpha$ in $\cH_{\alpha_j}$ whether the $\alpha$-subtree connected to $\alpha$ using $g$ belongs to $\Near(x,y)$ or not.
    Therefore, we can determine in $O(1)$ time that $\pi_{\alpha_j}(v)$ falls in exactly which of the three configurations (I-III) from \Cref{thm: characterization of skeleton subtrees}. 
    This immediately gives the classification of $v$ with respect to $\Near(x,y)$ in configurations (I) and (II). 
    In configuration (I), either $\pi_{\alpha_j}(v)$ is fully contained in an in-$\alpha$-subtree or contains two in-edges of $\alpha$. In both cases, we classify $v\in \Near(x,y)$ by \Cref{lem: two in edges,lem: fully contained in case}.
    In configuration (II), $\pi_{\alpha_j}(v)$ has one endpoint in an out-$\alpha$-subtree. Here, $v\notin \Near(x,y)$ by \Cref{lem: projection touches out edge}.
    Finally, it remains to consider the case when $\pi_{\alpha_j}(v)$ contains exactly one in-edge of $\alpha$, which is configuration (III).

    In this case, let $w,z$ be the stored chain-maximizers in the nearest mincut hierarchy. If no such chain-maximizers exist, then $M_v^\alpha=1$ by our stored information, and we classify $v\in \Near(x,y)$.
    One can determine whether $z\in \Near(x,y)$ and $w\in \Near(x,y)$ in $O(1)$ time since $w,z$ are terminals in $\alpha$, and $\Near(x,y)\cap \alpha$ is already known.
    Hence, the conditions of \Cref{lem: case ii configuration III z is in or wz out}(1,3) can be verified in $O(1)$ time: if $z\in \Near(x,y)$ or $w,z\notin \Near(x,y)$.
    In the former case, $v\in \Near(x,y)$ and in the latter case $v\notin \Near(x,y)$ by \Cref{lem: case ii configuration III z is in or wz out}(1,3). 
    As a special case, note that if $j=1$, then $z\in \Near(x,y)$ since $z$ is a terminal in $\alpha_1$ and $\alpha_1\subseteq\Near(x,y)$,
    Therefore, it remains to consider the case when $z\notin \Near(x,y)$ and $w\in \Near(x,y)$.
    In this case, we established in \Cref{lem: case ii configuration III z is in or wz out}(2) that $v\in \Near(x,y)$ if and only if $x\not\in \Near(w,z)$.
    To efficiently verify the latter condition, we utilize the following lemma, whose proof is deferred for readability.
    \begin{lemma} \label{lem: xy are terminals in the lca of wz}
        If $\set{w,z}\cap \Near(x,y)=\set{w}$, then $\LCA_{\cT}(x,y)$ is a descendant of $\LCA_{\cT}(w,z)$ in the nearest mincut hierarchy.
    \end{lemma}
    Denote $\LCA_{\cT}(w,z)$ by $\delta$.
    We show that $\delta\in \set{\alpha_1,\ldots,\alpha_j}$. 
    Firstly, $\delta$ must be a descendant of $\alpha_j$ as $w,z\in \alpha_j$.
    In addition, $\delta$ must be an ancestor of $\alpha_1$ by \Cref{lem: xy are terminals in the lca of wz}.
    Hence, $\delta$ lies on the path from $\LCA_{\cT}(x,y)=\alpha_1$ to $\alpha_j$ in the nearest mincut hierarchy. 
    Therefore, $x,y$ are vertices belonging to $\delta$.
    This allows us to determine whether $x\in \Near(w,z)$ in $O(1)$ time using the skeleton $\cH_{\delta}$, which is a standard application of $\LCA$ data structures stored for $\cH_\delta$ (e.g., see Lemma 2.15 in \cite{DBLP:conf/soda/BaswanaP22}).
    This concludes the proof of \Cref{lem: near from one level to the next}.
\end{proof}
We now establish \Cref{lem: xy are terminals in the lca of wz}.
\begin{proof}[Proof of \Cref{lem: xy are terminals in the lca of wz}]
    Denote $\delta:=\LCA(w,z)$. 
    Observe that $\Near(x,y)$ is an $\alpha_1$-mincut that separates terminals $w,z\in \gamma$ since $\set{w,z}\cap \Near(x,y)=\set{w}$, So, the capacity of any $\delta$-mincut is at most the capacity of $\alpha_1$-mincut. 
    Therefore, by construction of $\cT$ in the nearest mincut hierarchy, $\delta$ is either an ancestor of $\alpha_1$ ($\alpha_1=\delta$ is possible) or they are incomparable in $\cT$ (that is, $\delta$ is not on the path from $\alpha_1$ to the root $V$ of $\cT$).
    However, the latter is impossible since by \Cref{lem: subtree classification}, $\Near(x,y)$ can separate any subtree (and in particular it cannot separate any incomparable node) of the Skeletons from the ancestors of $\alpha_1$ in $\cT$, but $\Near(x,y)$ separates $w,z\in \gamma$.
\end{proof}
We conclude this section by establishing \Cref{lem: nearest mincut hierarchy query} using \Cref{lem: near from one level to the next} as follows.

\begin{proof}[Proof of \Cref{lem: nearest mincut hierarchy query}]
    Begin by finding $\Near(x,y)$ in the skeleton $\cH_{\alpha_1}$, which takes time $O(|\cH_{\alpha_1}|)$.
    Then, for each node $\mu\in \cH_{\alpha_1}$ that is in $\Near(x,y)$, mark the vertices in $\mu$ as being in $\Near(x,y)$.
    This takes time $O(|\alpha_1|)$ since we spend $O(1)$ time for each vertex in $\alpha_1$.
    Then, for each $i\in \{2,\ldots,k\}$, we apply \Cref{lem: near from one level to the next} to find $\Near(x,y)\cap \alpha_i$ from $\Near(x,y)\cap \alpha_{i-1}$ in time $O(|\alpha_i\setminus \alpha_{i-1}|)$.
    Since $\sum_{i=1}^k |\alpha_i\setminus \alpha_{i-1}|=O(n)$, the total running time of the algorithm is $O(n)$.
\end{proof}

\section{Reporting the Nearest Mincut Tree of a Vertex}
\label{sec:nearest-mincut-tree-construction-query}

In this section we prove \Cref{thm : nearest mincut tree reporting}.
To do so, we prove the following theorem, which combined with \Cref{lem: nearest mincut hierarchy space} immediately yields \Cref{thm : nearest mincut tree reporting}.
\begin{theorem}
\label{theorem:nearest-mincut-tree-construction-query}
    There exists an algorithm that, given a nearest mincut hierarchy of a graph $G=(V,E)$ and a source vertex $s\in V$, constructs the nearest mincut tree of $s$ in $O(n)$ time.
\end{theorem} 

The tree construction query works similarly to our nearest mincut reporting query (\Cref{lem: nearest mincut hierarchy query}), by recursively ascending the tree, starting from the parent of $s$.
Denoting the leaf to root path from $s$ to $V$ by $s=\alpha_1,\ldots,\alpha_k=V$, the algorithm maintains at every iteration $i\in [k]$ the nearest mincut tree of $s$ restricted to the vertices of $\alpha_i$.
Denote this partial tree by $T_s[\alpha_i]$, and note that the node containing $s$ in $\alpha_i$ is $\alpha_{i-1}$.
An illustration of the path from $s$ to $V$ in the nearest mincut hierarchy is given in \Cref{fig:tree-construction-path}.
Throughout, let $\Children(\alpha_i)$ be the set of children of $\alpha_i$ in $\cT$ of nearest mincut hierarchy, and denote $W\coloneqq \bigcup_{i=1}^k \left( \Children(\alpha_i)\setminus \set{\alpha_{i-1}} \right)$.
Our first result below shows that the nodes forming the nearest mincut tree of $s$ are exactly all the children of $\set{\alpha_i}_{i=1}^k$ in the nearest mincut hierarchy, minus the nodes $\set{\alpha_i}_{i=1}^k$ themselves.  

\begin{figure}[htbp]
    \centering
    \begin{tikzpicture}[
  vertex/.style={circle, draw, inner sep=2pt, minimum size=0.85cm},
  leaf/.style={circle, draw, inner sep=1pt, minimum size=0.5cm}
]
 
\node[vertex] (V) at (3.5, 5.5) {$\alpha_k {=} V$};
\node[vertex] (B) at (2.0, 3.5) {$\alpha_{k-1}$};
\node[vertex] (C) at (0.5, 1.5) {$\alpha_1$};
 
\draw (V) -- (B);
 
\foreach \t in {0.3, 0.5, 0.7}{
  \fill ($(B)!\t!(C)$) circle (1.5pt);
}
 
\node[leaf] (lv1) at (4.7, 6.1) {};
\node[leaf] (lv2) at (5.0, 5.5) {};
\node[leaf] (lv3) at (4.7, 4.9) {};
\draw (V) -- (lv1);
\draw (V) -- (lv2);
\draw (V) -- (lv3);
 
\node[leaf] (lb1) at (3.2, 4.1) {};
\node[leaf] (lb2) at (3.5, 3.5) {};
\node[leaf] (lb3) at (3.2, 2.9) {};
\draw (B) -- (lb1);
\draw (B) -- (lb2);
\draw (B) -- (lb3);
 
\node[leaf] (lc1) at (1.7, 2.1) {};
\node[leaf] (lc2) at (2.0, 1.5) {};
\node[vertex] (s)  at (-0.8, 0.3) {$s$};
\draw (C) -- (lc1);
\draw (C) -- (lc2);
\draw (C) -- (s);
 
\end{tikzpicture}
    \caption{Tree construction path.}
    \label{fig:tree-construction-path}
\end{figure}
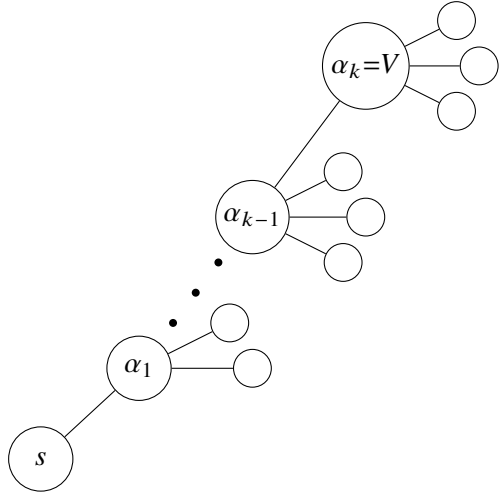

\begin{lemma}
    \label{lemma:nearest-mincut-vertices}
    The set of nodes of $T_s[V]$ is $W$. 
\end{lemma}
\begin{proof}
    We first argue that no node in $W$ is split in the nearest mincut tree of $s$.
    Let $w\in W$ be some node  such that $w\in \Children(\alpha_i)$, and fix two vertices $x,y\in w$.
    We show that $\Near(x,s)=\Near(y,s)$, which implies that $x,y$ are not separated by any nearest minimum cut to $s$, and in particular they are in the same node in the nearest mincut tree of $s$.
    Observe that since $\LCA(x,y)$ is a descendant of $\alpha_i$ we have by \Cref{claim:mincut-value-increases} that $\lambda_{x,y}>\lambda_{x,s}$, since $x,y$ are separated from $s$ by an $\alpha_i$-mincut.
    Therefore, any mincut separating $x$ from $s$ includes $y$, and any mincut separating $y$ from $s$ includes $x$.
    Hence, $\Near(x,s)\cap\Near(y,s)\ne \emptyset$, and by the laminarity of nearest mincuts we find $\Near(x,s)=\Near(y,s)$.

    We now show that every node in $W$ is not merged with any other node in the nearest mincut tree of $s$.
    Notice that any child of $\alpha_i$ and any child of $\alpha_j$ for $i\neq j$ are not merged in the tree, as the value of their mincut to $s$ is different by \Cref{claim:mincut-value-increases}.

    We conclude by showing that for every $i\in [k]$, no two children of $\alpha_i$ are merged in the $T_s[V]$.
    Fix two children $w_1,w_2$ of $\alpha_i$ in the nearest mincut hierarchy, and let $x\in w_1,y\in w_2$.
    Notice that $x,y$ are in the same node in $T_s[V]$ if and only if $\Near(x,s)=\Near(y,s)$.
    Notice that since $w_1,w_2$ are different nodes, there exists some $\alpha_i$-mincut $C\subseteq V$ that separates them.
    Assume without loss of generality that $x\in C$ and $s\not\in C$ (the cut must include exactly one of $x,y$ and might include $s$; if it includes $s$ then swap the roles of $x,y$).
    Therefore, $\Near(x,s)$ does not contain $y$, by the laminarity of nearest mincuts.
    Therefore, $\Near(x,s)\ne \Near(y,s)$, and $x,y$ are in different nodes in $T_s[V]$.
\end{proof}

To capture the structure of the nearest mincuts in relation to $s$ in each skeleton, we construct a \emph{local nearest mincut tree} $\Ltree_{\alpha_i}^s$ of $s$ in $\cH_{\alpha_i}$, which is defined as follows.
Begin by rooting $\cH_{\alpha_i}$ at $\alpha_{i-1}$.
Then, merge every empty node $q$ with its parent.
When constructing the tree, we mark for every non-empty node $q$ in $\cH_{\alpha_i}$ the non-empty node to which it is merged.
The resulting tree is the nearest mincut tree of $\alpha_{i-1}$, limited to the vertices of $\alpha_i$.
To make this more precise, we define the notion of a nearest mincut between disjoint vertex sets.
Given vertex sets $A,B\subseteq V$, we define $\Near(A,B)$ to be the nearest minimum cut separating $A$ from $B$, i.e., when contracting $A$ and $B$ into two supervertices, $\Near(A,B)$ is the nearest minimum cut separating these two supervertices.
\begin{claim}
    \label{claim:local-nearest-minimum-cut-tree}
    The resulting tree $\Ltree_{\alpha_i}^s$ is the nearest mincut tree of $\alpha_{i-1}$ limited to the vertices of $\alpha_i$, i.e. $T_{\alpha_{i-1}}[\alpha_i]$.
\end{claim}


\begin{proof}[Proof of \Cref{claim:local-nearest-minimum-cut-tree}]
The proof relies on the following observation, which is to be used later as well.
\begin{observation}
    \label{claim:same-min-cut-skeleton}
    Let $a\in \alpha_i$ be some vertex such that $a\not\in \alpha_{i-1}$, then $\Near(a,s)=\Near(a,\alpha_{i-1})$.
\end{observation}
\begin{proof}
    $\Near(a,s)$ is an $\alpha_i$-minimum cut as $a,s$ are in different nodes of $\cH_{\alpha_i}$.
    Similarly, $\Near(a,\alpha_{i-1})$ is an $\alpha_{i}$-minimum cut, since $a\in \alpha_i\setminus \alpha_{i-1}$.
    Therefore, $\lambda_{a,s}=\lambda_{a,\alpha_{i-1}}$.
    Notice that $a\in \Near(a,s)\cap\Near(a,\alpha_{i-1})$, therefore $\Near(a,s)\cap\Near(a,\alpha_{i-1})$ is an $a,\alpha_{i-1}$-cut and $\Near(a,s)\cup\Near(a,\alpha_{i-1})$ is an $\alpha_i$-cut as $s\in \alpha_{i-1}$ (and $\Near(a,s)\cup\Near(a,\alpha_{i-1})$ is a union of two $\alpha_i$-mincuts so they respect the nodes of $\cH_{\alpha_i}$).
    By submodularity (\Cref{lem:sub-posi-general}) we find that $\Near(a,s)\cap\Near(a,\alpha_{i-1})$ is an $a,\alpha_{i-1}$-minimum cut and hence $\Near(a,\alpha_{i-1})\subseteq \Near(a,s)$, and since $\lambda_{a,s}=\lambda_{a,\alpha_{i-1}}$ it follows that $\Near(a,s)=\Near(a,\alpha_{i-1})$.
\end{proof}
    By \Cref{claim:same-min-cut-skeleton}, for every vertex $a\in \alpha_i\setminus \alpha_{i-1}$, we have $\Near(a,s)=\Near(a,\alpha_{i-1})$.
    Hence, all the nearest minimum cuts of $\alpha_{i-1}$ in $\cH_{\alpha_i}$ are exactly the nearest minimum cuts of $s$ in $\cH_{\alpha_i}$. 
    Combining this with \Cref{lemma:nearest-mincut-vertices}, we find that the nodes of the nearest mincut tree of $\alpha_{i-1}$ in $\cH_{\alpha_i}$ are the children of $\alpha_{i-1}$ in the nearest mincut hierarchy, which are exactly the nodes of $\Ltree_{\alpha_i}^s$. 
    It remains to show that the cut structure of $\Ltree_{\alpha_i}^s$ is the same as the cut structure of the nearest mincut tree of $\alpha_{i-1}$ in $\cH_{\alpha_i}$.

    Let $x\in \alpha_i\setminus \alpha_{i-1}$ be a vertex and let $A=\Near(x,\alpha_{i-1})$ be its nearest minimum cut to $\alpha_{i-1}$ in $\cH_{\alpha_i}$.
    We will abuse the notation and denote by $x$ also the node of $\cH_{\alpha_i}$ containing $x$.
    Notice that for every non-empty ancestor of $y$ in $\cH_{\alpha_i}$, there exists a minimum $x,\alpha_{i-1}$-cut separating $x$ from $y$, defined by the edge outgoing from $x$ in the direction of $y$.
    Therefore, $y\not\in \Near(x,\alpha_{i-1})$.
    By a similar argument, every non-empty descendant $y$ of $x$ in $\cH_{\alpha_i}$, satisfies $y\in \Near(x,\alpha_{i-1})$.
    We conclude that $\Near(x,\alpha_{i-1})$ is the set of descendants of $x$ in $\cH_{\alpha_i}$, and the cut structure of $\Ltree_{\alpha_i}^s$ is the same as the cut structure of the nearest mincut tree of $\alpha_{i-1}$ in $\cH_{\alpha_i}$.
\end{proof}

Recall that by \Cref{claim:mincut-value-increases}, the value of $\alpha_i$-mincuts decreases as $i$ increases.
Therefore, any nodes in $\Children(\alpha_{i+1})$ cannot be an ancestor of any nodes in $\cup_{j\le i} \Children(\alpha_j)$ in the nearest mincut tree of $s$.
Combining this observation with \Cref{lemma:nearest-mincut-vertices}, we find that it is possible to construct the nearest mincut tree of $s$ by iteratively constructing the nearest mincut tree of $s$ restricted to $\alpha_i$ for $i=1,\ldots,k$.
This is achieved by combining the tree $T_s[\alpha_i]$ with the local tree $\Ltree_{\alpha_{i+1}}^s$.
The following lemma formalizes this intuition.
\begin{lemma}
    \label{lemma:all-subtree-same-parent}
    Let $\nu\in \Children(\alpha_{i+1})$ be a child of $\alpha_{i}$ in $\Ltree_{\alpha_{i+1}}^s$ and let $a\in \Children(\alpha_{i+1})$ be a node in the subtree of $\nu$ (in $\Ltree_{\alpha_{i+1}}^s$) that is not a child of $\alpha_{i}$ (i.e. $a\ne \nu$).
    Denote the parent of $a$ in $\Ltree_{\alpha_{i+1}}^s$ by $p(a)$, and its parent in $T_s[V]$ by $p'(a)$.
    Then, $p(a)=p'(a)$.
\end{lemma}
\begin{proof}
    Observe that $p(a),p'(a)$ are both nodes in $W$ by \Cref{lemma:nearest-mincut-vertices}.
    Assume towards contradiction that $p(a)\ne p'(a)$.
    Notice that $a\in \Near(p(a),\alpha_{i})= \Near(p(a),s)$ by \Cref{claim:same-min-cut-skeleton}.
    In addition, $a\in \Near(p'(a),s)$ as $a$ is a descendant of $p'(a)$ in $T_s[V]$.
    Therefore, by the laminarity of nearest minimum cuts to a fixed source, either $\Near(p'(a),s)\subseteq \Near(p(a),s)$ or $\Near(p(a),s)\subseteq \Near(p'(a),s)$.
    This implies that $p(a),p'(a)$ are related by ancestry in $T_s[V]$.
    
    We now split into two cases.
    The first is when $p'(a)$ is a strict ancestor of $p(a)$ in $T_s[V]$, and hence $\Near(p(a),s)\subseteq \Near(p'(a),s)$.
    However, this implies that $a$ is not a descendant of $p(a)$ in $T_s[V]$ as $a$ is a child of $p'(a)$ by our assumption;
    this contradicts the fact that $a\in \Near(p(a),s)$ as shown above.

    In the complementary case, we show that $p'(a)$ is in $\Children(\alpha_{i+1})\setminus\{\alpha_{i}\}$.
    Recall that $p(a)\ne \alpha_{i}$ by the lemma assumption.
    In addition, notice that $p'(a)\not\subseteq \alpha_{i}$ as $\Near(p(a),s)=\Near(p(a),\alpha_{i})$ by \Cref{claim:same-min-cut-skeleton} and $p'(a)\subseteq \Near(p(a),s)$.
    Furthermore, $p'(a)$ cannot be in $\Children(\alpha_{j})$ for $j>i+1$ since $\lambda_{a,s} \le \lambda_{p'(a),s}$ by the structure of the nearest mincut tree.
    By the lemma assumption we have $a\subseteq \alpha_{i+1}$, and hence $\Near(a,s)$ is an $\alpha_{i+1}$-mincut.
    However, by \Cref{claim:mincut-value-increases} the value of $\alpha_j$-mincuts is strictly smaller than the value of $\alpha_{i+1}$-mincuts for $j>i+1$.
    Therefore, $p'(a)$ is a child of $\alpha_{i+1}$ in $T_s[V]$ such that $p'(a)\ne \alpha_{i}$. 
    In particular, $p'(a)$ is a node in $\Ltree_{\alpha_{i+1}}^s$.

    This yields a contradiction following similar lines to the previous case, whether $p'(a)$ is an ancestor or descendant of $p(a)$ in $\Ltree_{\alpha_{i+1}}^s$.
    Similarly, if $p'(a)$ and $p(a)$ are not related by ancestry in $\Ltree_{\alpha_{i+1}}^s$ then we also reach a contradiction as $\Near(p(a),s)\cap \Near(p'(a),s)=\emptyset$ which contradicts the fact that both $p(a)$ and $p'(a)$ contain $a$ as argued above.
    Hence, we conclude that $p(a)=p'(a)$.
\end{proof}
Therefore, each tree of $\Ltree_{\alpha_{i+1}}^s\setminus \alpha_{i}$ is attached to $T_s[\alpha_{i}]$ by connecting its root to some node in $T_s[\alpha_{i}]$.
The next lemma shows how to find the node in $T_s[\alpha_{i}]$ to which the root of each tree of $\Ltree_{\alpha_{i+1}}^s\setminus \alpha_{i}$ is attached.
First, for every vertex $v\in \alpha_{i}$, let $N_i(v)=\set{w \in \alpha_{i} \mid v\in \Near(w,s)}$ be the set of vertices in $\alpha_{i}$ whose nearest minimum cut to $s$ contains $v$. 
By the structure of the nearest mincut tree $T_s[V]$, $N_i(v)$ corresponds to a path in $T_s[V]$.
We call the lowest node in this path the \emph{attachment node} of $v$ and denote it by $\Att(v)$.
If $N_i(v)=\emptyset$, then $\Att(v)=s$.
\begin{lemma}
    \label{lemma:attachment-node-is-parent}
    Fix some $i\in [1,k]$ and let $A$ be an $\alpha_{i}$-subtree in $\cH_{\alpha_{i+1}}$.
    Let $v$ be the representative terminal of $A$, and $\mu$ be a child of $\alpha_{i}$ in $\Ltree_{\alpha_{i+1}}^s$ such that $\mu\subseteq A$.
    Then, the parent of $\mu$ in $T_s[V]$ is $\Att(v)$.
\end{lemma}
\begin{proof}
    Begin by noting that for every $u \in \alpha_{i}$, we have that $\mu\subseteq \Near(u,s)$ if and only if $v\in \Near(u,s)$ by \Cref{lem: subtree classification}.
    Denote the parent of $\mu$ in $T_s[V]$ by $p(\mu)$.
    We first show that $p(\mu)\subseteq \alpha_{i}$.
    Since $\mu$ is a child of $\alpha_{i}$ in $\Ltree_{\alpha_{i+1}}^s$, we have that $p(\mu)\not\subseteq \alpha_{i+1}\setminus \alpha_{i}$.
    We now show that $p(\mu)$ is not a child of $\alpha_{j}$ in $T_s[V]$ for any $j>i+1$.
    If $p(\mu)$ is a child of $\alpha_{j}$ for some $j>i+1$, then by the properties of the nearest mincut hierarchy $\lambda_{p(\mu),s}< \lambda_{\mu,s}$ since $\mu$ is separated from $s$ by an $\alpha_{i+1}$-minimum cut and $p(\mu)$ by an $\alpha_{j}$-minimum cut and using \Cref{claim:mincut-value-increases}.
    Therefore, $p(\mu)$ is not an ancestor of $\mu$ in $T_s[V]$, which is a contradiction.
    We conclude that $p(\mu)\subseteq \alpha_{i}$.

    Now, observe that for every node $w\in W$ such that $w\subseteq \alpha_{i}$, we have that $v\in \Near(w^*,s)$, for some $w^*\in w$, if and only if $w\subseteq N_i(v)$.
    Therefore, if $N_i(v)=\emptyset$, then $\mu$ is not inside the nearest minimum cut of any node in $\alpha_{i}$, and hence $p(\mu)=s=\Att(v)$.
    For the rest of the proof assume that $N_i(v)\ne \emptyset$ and hence $\pi(\mu)\ne s$ since $\mu$ is included in some nearest minimum cut.
    Since $p(\mu)\subseteq \alpha_{i}$, we have that $v\in \Near(p(\mu),s)$, and hence there exists some $p^*\in p(\mu)$ such that $p^*\in N_i(v)$.
    In particular, $p(\mu)$ has to be the lowest node in $N_i(v)$.
    Otherwise, there exists some $p'\in N_i(v)$ such that $p'$ is a descendant of $p(\mu)$ in $T_s[V]$ and $\mu\subseteq \Near(p',s)$, which is a contradiction to $\mu$ being a child of $p(\mu)$ in $T_s[V]$.
    Hence, we conclude that $p(\mu)=\Att(v)$, as required.
\end{proof}
The following lemma is the main result of this section, and it implies \Cref{theorem:nearest-mincut-tree-construction-query} immediately.
\begin{lemma}
    \label{lemma:find-attachment-node-query}
    There exists an algorithm that, given a vertex $s\in V$, the projection $\pi_{\alpha_i}(v)$ of a vertex $v\in \alpha_{i+1}\setminus \alpha_i$, the nearest mincut hierarchy of $G$, a partial nearest mincut tree $T_s[\alpha_i]$, and an LCA oracle on $T_s[\alpha_i]$, finds $\Att(v)$ in $O(1)$ time.
\end{lemma}
\begin{proof}[Proof of \Cref{theorem:nearest-mincut-tree-construction-query}]
    Throughout maintain an LCA oracle on $T_s[\alpha_i]$ while ascending the tree.
    Observe that since the algorithm only inserts new nodes as children of existing nodes, the LCA oracle can be maintained in $O(1)$ time per insertion by using the dynamic LCA of \cite{DBLP:journals/siamcomp/ColeH05}.
    Begin by constructing the nearest mincut tree of $s$ restricted to $\alpha_2$.
    This is obtained in $O(|\cH_{\alpha_2}|)$ time by constructing the local nearest mincut tree of $s$ in $\cH_{\alpha_2}$.

    We now focus on the $i$-th iteration of the algorithm, where $T_s[\alpha_i]$ was already constructed, and the goal is to extend it to $T_s[\alpha_{i+1}]$.
    First, construct the local nearest mincut tree $\Ltree_{\alpha_{i+1}}^s$ of $s$ in $\cH_{\alpha_{i+1}}$ in $O(|\cH_{\alpha_{i+1}}|)$ time.
    Then, for every $\alpha_{i}$-subtree $A$ in $\cH_{\alpha_{i+1}}$, find its representative terminal $v$ and its attachment node $\Att(v)$ in $T_s[\alpha_i]$ using \Cref{lemma:find-attachment-node-query}.
    This takes $O(1)$ time per representative terminal, and hence $O(|{\cH_{\alpha_{i+1}}}|)$ time in total.
    Finally, attach each subtree of $\Ltree_{\alpha_{i+1}}^s$ to $T_s[\alpha_i]$ according to the attachment nodes and update the LCA oracle accordingly, which takes $O(|\cH_{\alpha_{i+1}}|)$ time.
    The total running time of the algorithm is $O(\sum_{i=1}^k |\cH_{\alpha_i}|)=O(n)$ since the total size of all skeletons is $O(n)$.
\end{proof}
\begin{proof}[Proof of \Cref{lemma:find-attachment-node-query}]
    Recall that given a vertex $v\in \alpha_{i+1}\setminus \alpha_i$, we denote by $N_i(v)$ the set of vertices in $\alpha_i$ whose nearest minimum cut to $s$ contains $v$.
    At a high level, the proof uses \Cref{thm: characterization of skeleton subtrees}, in an opposite manner to \Cref{lem: nearest mincut hierarchy query}.
    Instead of fixing a pair of vertices $x,y$ and finding for a given representative vertex $v$ whether $v\in\Near(x,y)$, we fix the representative vertex $v$ and wish to find the lowest $x$ in the nearest mincut tree of $s$ such that $v\in\Near(x,s)$.
    To do so, we leverage the insights of \Cref{thm: characterization of skeleton subtrees} which allow us to find the set $N_i(v)$ efficiently.

    For a non-empty node $u$ of $\cH_{\alpha_i}$ and any $x\in T_s[\alpha_i]$, we have $u\in\Near(x,s)$ if and only if $x$ is an ancestor of $u$ in $T_s[\alpha_i]$.
    Given several nodes $u_1,\dots,u_r$, the lowest $x$ satisfying $u_1,\dots,u_r\in\Near(x,s)$ is therefore $\LCA(u_1,\dots,u_r)$.
    The algorithm is split into three cases depending on $\pi_{\alpha_i}(v)$, the projection of $v$ in $\cH_{\alpha_i}$, illustrated in \Cref{fig:min-cut-tree-query-projection-cases}.

    \emph{Case (a): $\pi_{\alpha_i}(v)$ does not touch the node $\alpha_{i-1}$.}
    Notice that any $\cH_{\alpha_i}$ containing both endpoints $\nu_1,\nu_2$ of $\pi_{\alpha_i}(v)$ contains all of it. Hence,
    if $\LCA_{\cH_{\alpha_i}}(\nu_1,\nu_2)$ is a non-empty node, then $v\in\Near(x',s)$ for $x'\in x\coloneqq \LCA(\nu_1,\nu_2)$ and $v\not\in \Near(y',s)$ for every $y'\in y$ where $y$ is a non-empty proper descendant of $x$, so $\Att(v)=x$.
    Otherwise, $\LCA(\nu_1,\nu_2)$ is empty, and we need to find the lowest ancestor of $\LCA(\nu_1,\nu_2)$ in $\cH_{\alpha_i}$ that is non-empty.
    Denote this ancestor by $\mu$.
    If $\mu\ne\alpha_{i-1}$, then $v\in\Near(x',s)$ holds for any $x'\in \mu$, and $v\not\in \Near(y',s)$ for every $y'\in y$ where $y$ is a non-empty proper descendant of $\mu$, therefore $\Att(v)=\mu$.

    Otherwise, we show that $\Att(v)\in \alpha_{i-1}$.
    Assume towards contradiction that $\Att(v)\not\in \alpha_{i-1}$, and let $\tau$ be a node in $\cH_{\alpha_i}\setminus\{\alpha_{i-1}\}$ such that some vertex $t\in \tau$ satisfies $t\in N_i(v)$.
    Fix the $\alpha_{i}$-mincut $C$ separating $t$ from $s$ corresponding to some edge $g$ incident to $\tau$, such that $\alpha_{i-1}\subseteq \overline{C}$.
    Since $t\in N_i(v)$, we have that $v\in C$ and hence $\pi_{\alpha_i}(v)\subseteq C$.
    Therefore, $g$ must separate $\alpha_{i-1}$ from $\LCA(\nu_1,\nu_2)$ in $\cH_{\alpha_i}$, and $\tau$ is a non-empty node between $\alpha_{i-1}$ and $\LCA(\nu_1,\nu_2)$ which is a contradiction to our assumption.
    Since $N_i(v)\subseteq \alpha_{i-1}$, we have that $x\in N_i(v)$ if and only if $A\subseteq \Near(x,s)$, where $A$ is the $\alpha_{i-1}$-subtree containing $\pi_{\alpha_i}(v)$ by \Cref{thm: characterization of skeleton subtrees}.
    Therefore, $\Att(v)=\LCA(a_1,\ldots,a_r)$ in $T_s[\alpha_i]$, where $a_1,\ldots,a_r$ are the vertices in $A$.
    This is difficult to compute directly, but instead one can simply take the attachment node of representative terminal of $A$.

    To analyze the time complexity, notice that $\LCA(\nu_1,\nu_2)$ is determined in $O(1)$ time using the LCA oracle on $\cH_{\alpha_i}$.
    Then, to find its first non-empty ancestor of $\LCA(\nu_1,\nu_2)$ we take the node into which $\LCA(\nu_1,\nu_2)$ is merged in $\Ltree_{\alpha_i}^s$.
    This can be done in $O(1)$ since we save the nodes into which  every empty node was merged in the construction $\Ltree_{\alpha_i}^s$.
    Finally, we can find the attachment node of the representative terminal of $A$ since it is stored in the nearest mincut hierarchy, and then we already have its attachment node from the previous iteration of the algorithm.

    \emph{Case (b): $\pi_{\alpha_i}(v)$ meets two edges incident to $\alpha_{i-1}$.}
    Let $A_1,A_2$ be the $\alpha_{i-1}$-subtrees containing these edges.
    Notice that for every $x\in \alpha_i\setminus \alpha_{i-1}$, we have that $v\not\in \Near(x,s)$ since $\pi_{\alpha_i}(v)$ is not contained in any $\alpha_i$-mincut separating $x$ from $s$.
    Therefore, $\Att(v)\in \alpha_{i-1}$.
    By \Cref{thm: characterization of skeleton subtrees}, for every vertex $x\in \alpha_{i-1}$ we have $v\in\Near(x,s)$ if and only if $A_1,A_2\subseteq\Near(x,s)$.
    Let $u_1,u_2$ be non-empty nodes in $A_1,A_2$, respectively.
    Then, $v\in\Near(x,s)$ if and only if $u_1,u_2\in\Near(x,s)$, and hence $\Att(v)=\LCA(u_1,u_2)$ in $T_s[\alpha_i]$.
    The running time of this query is the time needed to find $u_1,u_2$ and compute their LCA in $T_s[\alpha_i]$.
    The first part can be computed in $O(1)$ time since it is stored in the nearest mincut hierarchy, and the second part can be computed in $O(1)$ time using the LCA oracle on $T_s[\alpha_i]$.

    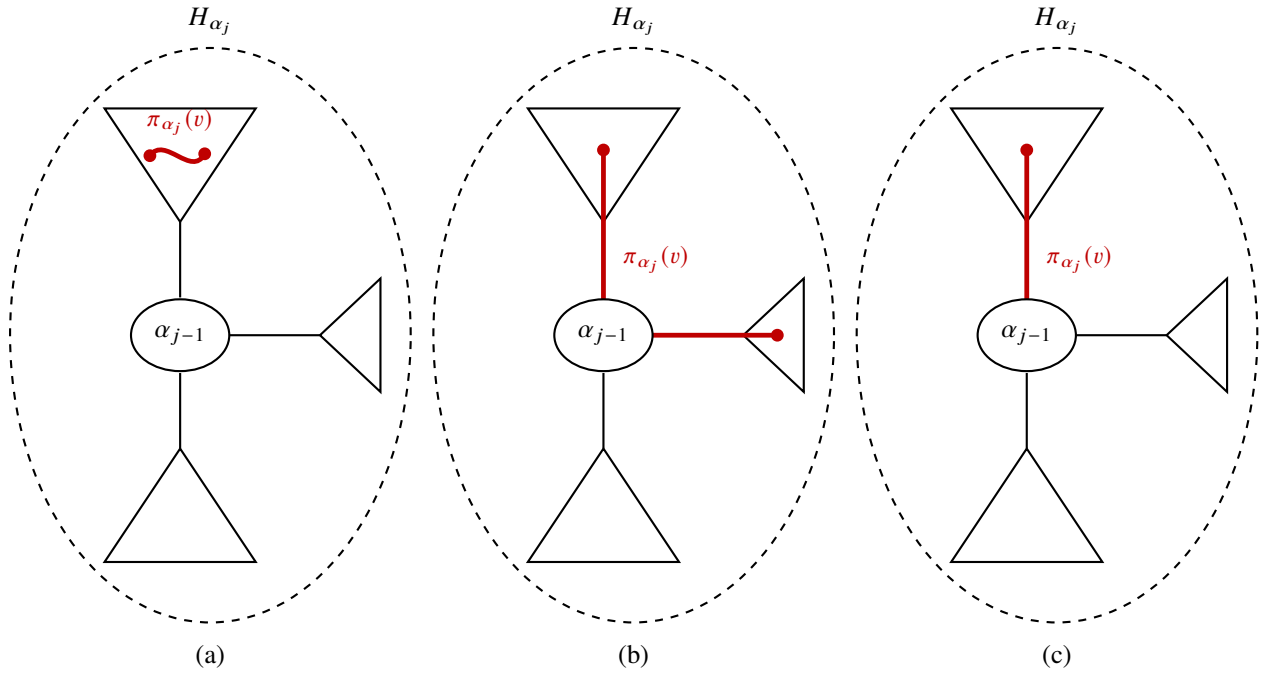
\begin{figure}[htbp]
        \centering
        \begin{tikzpicture}[
  thick,
  vertex/.style={
    ellipse, draw, fill=white,
    minimum width  = 1.30cm,
    minimum height = 0.95cm,
    inner sep = 2pt,
    font = \small,
  },
  proj/.style={
    draw = red!75!black, line width = 1.8pt, line cap = round, fill = none,
  },
]

\begin{scope}[xshift=0cm]

  \draw[dashed] (0.4,0) ellipse (2.65cm and 3.8cm);
  \node[font=\small] at (0.4, 4.15) {$H_{\alpha_j}$};

  \draw (0,  0.5) -- (0,  1.5);   
  \draw (0.65, 0) -- (1.85, 0);   
  \draw (0, -0.5) -- (0, -1.5);   

  \draw[fill=white] ( 0,  1.5) -- (-1,  3.0) -- (1,  3.0) -- cycle;      
  \draw[fill=white] (1.85, 0 ) -- (2.65,-0.75) -- (2.65, 0.75) -- cycle;  
  \draw[fill=white] ( 0, -1.5) -- (-1, -3.0) -- (1, -3.0) -- cycle;      

  \draw[proj]
    (-0.40, 2.375) .. controls (-0.20, 2.65) and (0.15, 2.075) .. (0.325, 2.40);
  \filldraw[red!75!black] (-0.40, 2.375) circle (2pt);
  \filldraw[red!75!black] ( 0.325, 2.40) circle (2pt);

  \node[vertex] at (0,0) {$\alpha_{j-1}$};

  \node[font=\footnotesize, text=red!75!black] at (0, 2.82)
    {$\pi_{\alpha_j}(v)$};

  \node[font=\small] at (0.4, -4.25) {(a)};
\end{scope}

\begin{scope}[xshift=5.6cm]

  \draw[dashed] (0.4,0) ellipse (2.65cm and 3.8cm);
  \node[font=\small] at (0.4,4.15) {$H_{\alpha_j}$};

  \draw (0,  0.5) -- (0,  1.5);
  \draw (0.65, 0) -- (1.85, 0);
  \draw (0, -0.5) -- (0, -1.5);

  \draw[fill=white] ( 0,  1.5) -- (-1,  3.0) -- (1,  3.0) -- cycle;
  \draw[fill=white] (1.85, 0 ) -- (2.65,-0.75) -- (2.65, 0.75) -- cycle;
  \draw[fill=white] ( 0, -1.5) -- (-1, -3.0) -- (1, -3.0) -- cycle;

  \draw[proj] (0, 2.45) -- (0, 0);
  \draw[proj] (0.65, 0) -- (2.30, 0);
  \filldraw[red!75!black] (0,    2.45) circle (2pt);
  \filldraw[red!75!black] (2.30, 0   ) circle (2pt);

  \node[vertex] at (0,0) {$\alpha_{j-1}$};

  \node[font=\footnotesize, text=red!75!black, anchor=west]
    at (0.12, 1.0) {$\pi_{\alpha_j}(v)$};

  \node[font=\small] at (0.4,-4.25) {(b)};
\end{scope}

\begin{scope}[xshift=11.2cm]

  \draw[dashed] (0.4,0) ellipse (2.65cm and 3.8cm);
  \node[font=\small] at (0.4,4.15) {$H_{\alpha_j}$};

  \draw (0,  0.5) -- (0,  1.5);
  \draw (0.65, 0) -- (1.85, 0);
  \draw (0, -0.5) -- (0, -1.5);

  \draw[fill=white] ( 0,  1.5) -- (-1,  3.0) -- (1,  3.0) -- cycle;
  \draw[fill=white] (1.85, 0 ) -- (2.65,-0.75) -- (2.65, 0.75) -- cycle;
  \draw[fill=white] ( 0, -1.5) -- (-1, -3.0) -- (1, -3.0) -- cycle;

  \draw[proj] (0, 2.45) -- (0, 0);
  \filldraw[red!75!black] (0, 2.45) circle (2pt);

  \node[vertex] at (0,0) {$\alpha_{j-1}$};

  \node[font=\footnotesize, text=red!75!black, anchor=west]
    at (0.12, 1.0) {$\pi_{\alpha_j}(v)$};

  \node[font=\small] at (0.4,-4.25) {(c)};
\end{scope}

\end{tikzpicture}
        \caption{The three projection configurations of $\pi_{\alpha_i}(v)$ relative to $\alpha_{i-1}$ used by the attachment-node routine: (a) fully contained in one subtree, (b) traversing two incident edges, and (c) traversing a single incident edge.}
        \label{fig:min-cut-tree-query-projection-cases}
    \end{figure}

    \emph{Case (c): $\pi_{\alpha_i}(v)$ meets a single edge $e$ incident to $\alpha_{i-1}$.}
    Denote the $\alpha_{i-1}$-subtree incident to $e$ by $A$.
    Let $w,z$ be the chain maximizers of $v$ (\Cref{thm : chain-maximizer}), and identify $w,z$ with their nodes in $\cH_{\alpha_i}$.
    We first argue that $\Att(v)\in \alpha_{i-1}$.
    Every vertex $y\in \alpha_i\setminus\alpha_{i-1}$ is separated from $s$ by an $\alpha_i$-mincut say, $C$, corresponding to the valid partition $A,\alpha_{i-1}\setminus A$ defined by $e$ such that $v$ is on side of $\alpha_{i-1}\setminus A$ for $C$.
    Notice that $C$ does not contain $e$, and hence $\pi_{\alpha_i}(v)\not\subseteq C$ and $v\not\in\Near(y,s)$.

    Now, fix some vertex $x\in \alpha_{i-1}$, and let $u$ be a non-empty node with $u\in A$.
    By \Cref{thm: characterization of skeleton subtrees}, if $A$ is an out-subtree, then $v\notin\Near(x,s)$; and by \Cref{lem: subtree classification} this happens if and only if $u\not\in\Near(x,s)$.
    Since only ancestors of $u$ contain it in their nearest mincut to $s$, we find that $\Att(v)$ lies on the path from $s$ to $u$ in $T_s[\alpha_i]$.
    By \Cref{thm: characterization of skeleton subtrees}, $z\in\Near(x,s)$ implies $v\in\Near(x,s)$, while $w,z\notin\Near(x,s)$ implies $v\notin\Near(x,s)$.
    Hence, if $\LCA_{T_s[\alpha_i]}(z,u)$ is a descendant of $\LCA_{T_s[\alpha_i]}(w,u)$ then $\Att(v)=\LCA_{T_s[\alpha_i]}(z,u)$.
    Otherwise, $\Att(v)$ lies on the (inclusive) path $P$ from $\LCA_{T_s[\alpha_i]}(u,w)$ to $\LCA_{T_s[\alpha_i]}(u,z)$.
    An illustration of this case is given in \Cref{fig:attachment-vertex-case-c}.
    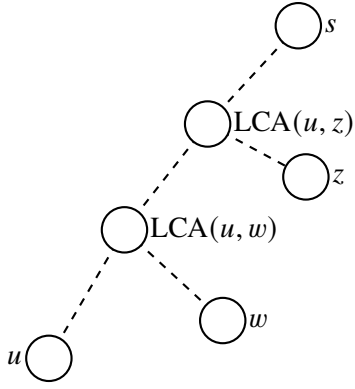
\begin{figure}
        \centering
        \begin{tikzpicture}[vtx/.style={circle, draw, thick, fill=white, minimum size=6mm, inner sep=0pt}]

    \coordinate (s)     at (4.2, 5.0);
    \coordinate (lcauz) at (3.0, 3.7);
    \coordinate (lcauw) at (1.9, 2.3);
    \coordinate (u)     at (0.9, 0.6);
    \coordinate (z)     at (4.3, 3.0);
    \coordinate (w)     at (3.2, 1.1);

    \draw[dashed, thick] (s) -- (lcauz);
    \draw[dashed, thick] (lcauz) -- (lcauw);
    \draw[dashed, thick] (lcauw) -- (u);
    \draw[dashed, thick] (lcauz) -- (z);
    \draw[dashed, thick] (lcauw) -- (w);

    \node[vtx] at (s) {};
    \node[vtx] at (lcauz) {};
    \node[vtx] at (lcauw) {};
    \node[vtx] at (u) {};
    \node[vtx] at (z) {};
    \node[vtx] at (w) {};

    \node[anchor=west, xshift=6pt]  at (s)     {$s$};
    \node[anchor=west, xshift=6pt]  at (lcauz) {$\LCA(u,z)$};
    \node[anchor=west, xshift=6pt]  at (lcauw) {$\LCA(u,w)$};
    \node[anchor=east, xshift=-6pt] at (u)     {$u$};
    \node[anchor=west, xshift=6pt]  at (z)     {$z$};
    \node[anchor=west, xshift=6pt]  at (w)     {$w$};

\end{tikzpicture}
        \caption{Illustration of the path $P$ from $\LCA_{T_s[\alpha_i]}(u,w)$ to $\LCA_{T_s[\alpha_i]}(u,z)$ in $T_s[\alpha_i]$.}
        \label{fig:attachment-vertex-case-c}
    \end{figure}

    By \Cref{thm: characterization of skeleton subtrees}, $v\in\Near(x,s)$ if and only if $x\notin\Near(w,z)$.
    Therefore, it remains to classify if $x\in\Near(w,z)$ for every $x$ on $P$.
    Let $x'$ be the first node on $P$ with $x'\ne w$, i.e. $x'=\LCA_{T_s[\alpha_i]}(u,w)$ if $\LCA_{T_s[\alpha_i]}(u,w)\ne w$ and $x'$ is the parent of $\LCA_{T_s[\alpha_i]}(u,w)$ otherwise.
    We show that $\Att(v)=x'$ if $\lambda_{w,s}\le\lambda_{z,s}$, and $\Att(v)=\LCA_{T_s[\alpha_i]}(u,z)$ otherwise.

    Start by considering the case $\lambda_{w,s}\le\lambda_{z,s}$.
    Observe that $\Near(w,s)\cap\Near(w,z)$ is a $w,s$-cut and $\Near(w,s)\cup\Near(w,z)$ is a $w,z$-cut.
    By submodularity (\Cref{lem:sub-posi-general}) we have that $\Near(w,s)\cap\Near(w,z)$ is a $w,s$-mincut, and by the minimality of $\Near(w,s)$ we have that $\Near(w,s)\subseteq \Near(w,z)$.
    Furthermore, if $s\not\in \Near(w,z)$ we can apply the same argument to show that $\Near(w,z)\subseteq \Near(w,s)$ and hence $\Near(w,z)=\Near(w,s)$.
    If this holds then every strict ancestor of $w$ is not contained in $\Near(w,z)$ and hence $x'\not\in \Near(w,z)$.
    Since $x'$ is the first node on $P$ for which this applies then $\Att(v)=x'$.
    Assume towards contradiction that $s\in \Near(w,z)$.
    Then, $\lambda_{w,z}\ge\lambda_{z,s}$ as $\Near(w,z)$ is a $w,z$-minimum cut and a $z,s$-cut.
    In the other direction, $z\not\in \Near(w,s)$ and hence $ \lambda_{z,s}\ge \lambda_{w,s} \ge \lambda_{w,z}$, where the first inequality is from the case assumption, and the second is since $\Near(w,s)$ is a $w,s$-minimum cut and a $w,z$-cut.
    Therefore, $\lambda_{z,s}\ge \lambda_{w,s} \ge \lambda_{w,z} \ge \lambda_{z,s}$ and $\lambda_{w,s} = \lambda_{w,z} = \lambda_{z,s}$.
    This leads to a contradiction, since we found a minimum cut separating $w,z$ without $s$ in it (namely $\Near(w,s)$), which is a contradiction to the minimality of $\Near(w,z)$.

    In the complementary case, $\lambda_{w,s} > \lambda_{z,s}$.
    This implies that $s\in \Near(w,z)$ as if there exists any $w,s$-minimum cut that separates $w$ and $z$ it would give $\lambda_{z,s}=\lambda_{w,z}=\lambda_{w,s}$, contradicting the fact that $\lambda_{w,s} > \lambda_{z,s}$.
    Examine any node $y$ on $P$ that is strictly below $\LCA_{T_s[\alpha_i]}(u,z)$ in $T_s[\alpha_i]$.
    Assume towards contradiction that $y\not\in \Near(w,z)$.
    Notice that the cuts $\Near(w,z),\Near(y,s)$ satisfy the conditions of \Cref{lem: crossing mincuts with special assignment of vertex pairs} with $w=a,z=b,y=a',s=b'$.
    Therefore, $\Near(w,z)\cap \Near(y,s)$ is a strictly smaller $w,z$-minimum cut than $\Near(w,z)$, which is a contradiction.
    This immediately implies that for every node $y$ on the path from $\LCA_{T_s[\alpha_i]}(u,w)$ to $\LCA_{T_s[\alpha_i]}(u,z)$ we have $y\in\Near(w,z)$ and hence $v\not\in \Near(y,s)$.
    Finally, we conclude that $\Att(v)=\LCA_{T_s[\alpha_i]}(u,z)$.

    To conclude the proof, we analyze the time complexity of this case.
    First, we find $u$ in $O(1)$ time by taking some non-empty node in $A$, which we can find using the end of $\pi_{\alpha_i}(v)$ that is not $\alpha_{i-1}$.
    Then, find $\LCA_{T_s[\alpha_i]}(u,w),\LCA_{T_s[\alpha_i]}(u,z)$ in $O(1)$ time using the LCA oracle on $T_s[\alpha_i]$.
    To find $\lambda_{w,s},\lambda_{z,s}$ we can check the value of the edges outgoing from $w,z$ in $T_s[\alpha_i]$, which can be done in $O(1)$ time.
    This concludes the proof of the lemma.
\end{proof}


\section{Sensitivity Oracles for the Insertion of an Edge}
\label{sec:sensitivity-oracles-insertion-edge}
In this section we prove our applications of the nearest mincut hierarchy from \Cref{thm:main} to sensitivity oracles for all-pairs mincuts under the insertion of an edge. 
In particular, we prove \Cref{cor:all-pairs-insertion-sensitivity-oracle} and \Cref{thm : mincut sensitivity data structures for the insertion of an edge}.

\paragraph{Single-pair query.} 
In this setting, a query is formed of a pair of vertices $u,v\in V$ and an edge $e=(x,y)$ with any arbitrary weight, and the goal is to determine if the insertion of $e$ increases the value of the $(u,v)$-mincut.
Our proof is based on the following fact from \cite{DBLP:journals/mp/PicardQ80}.
\begin{fact} [\cite{DBLP:journals/mp/PicardQ80}] \label{fact: insertion-single-pair-query}
    For any pair of vertices $u,v$, the value of $(u,v)$-mincut increases upon insertion of edge $e=(x,y)$ if and only if $x\in \Near(u,v)$ and $y\in \Near(v,u)$ or vice versa.
\end{fact}
\begin{proof}[Proof of \Cref{cor:all-pairs-insertion-sensitivity-oracle}]
It follows from \Cref{thm:main} that the nearest mincut hierarchy can report $\Near(u,v)$ and $\Near(v,u)$ in $O(n)$ time.
Therefore, by \Cref{fact: insertion-single-pair-query}, we need to just verify whether $x\in \Near(u,v)$ and $y\in \Near(v,u)$ or vice versa. This would take an additional $O(1)$ time only. 
This completes the proof of \Cref{cor:all-pairs-insertion-sensitivity-oracle}.
\end{proof}

Consider the special case where the given graph happens to be integrally weighted and the query edge $e$ is of unit weight. 
In this case, observe that at least one of $\Near(u,v)$ and $\overline{\Near(v,u)}$ is a $(u,v)$-mincut in the resulting graph upon insertion of $e$. 
This is because if, for edge $e$, the latter condition from \Cref{fact: insertion-single-pair-query} is not satisfied, then at least one of $\Near(u,v)$ and $\Near(v,u)$ remains a $(u,v)$-mincut; otherwise, both become $(u,v)$-mincut since the value of mincut increases either by $1$ or remains the same upon insertion of unit weight edge. 
Therefore, our nearest mincut hierarchy (\Cref{thm:main}) is also capable of reporting one $(u,v)$-mincut in this special setting.

\paragraph{All-pairs query.} 
In the all-pairs query setting, the query is formed of an edge $e=(x,y)$ with any arbitrary weight, and the goal is to determine for every pair of vertices $u,v\in V$, whether the insertion of $e$ increases the value of the $(u,v)$-mincut.
A naive approach, leveraging \Cref{fact: insertion-single-pair-query}, is to explicitly query every pair of vertices $u,v\in V$ using the nearest mincut hierarchy (\Cref{thm:main}) to report $\Near(u,v)$ and $\Near(v,u)$, and then check if $x\in \Near(u,v)$ and $y\in \Near(v,u)$ or vice versa.
However, this approach would take $O(n^3)$ time, which is not efficient.
In order to achieve a faster query time, we exploit the following result of \cite{BaswanaGK22}. 
\begin{theorem} [Section 5.2 in \cite{BaswanaGK22}] \label{thm: all-pairs insertion of baswana et al}
    For any undirected weighted graph $G=(V,E,w)$ on $n$ vertices, given the nearest mincut tree of  every vertex $s\in V$, one can report in $O(k)$ time all the $k$ pairs of vertices for which the mincut value increases upon insertion of any edge $e$.
\end{theorem}
\begin{proof}[Proof of \Cref{thm : mincut sensitivity data structures for the insertion of an edge}]
    Begin by constructing the nearest mincut tree for every vertex $s$ using \Cref{thm : nearest mincut tree reporting}. 
    Then, use \Cref{thm: all-pairs insertion of baswana et al} to report all the pairs of vertices for which the mincut value increases upon insertion of any edge $e$ in $O(k)$ time, where $k$ is the number of such pairs.
    The total time taken to report all the pairs is $O(n^2+k)$, which is $O(1)$ per vertex pair.
\end{proof}

\section{Minimal Gomory-Hu Tree Lower Bound}
\label{sec:nearest-gomory-hu-lower-bound}

The minimal Gomory-Hu tree for a graph is defined as follows.
\begin{theorem}[\cite{DBLP:conf/focs/AbboudKLPGSYY25}] \label{thm: minimal Gomory Hu Tree}
    Let $G=(V,E)$ be an undirected weighted graph and let $s \in V$.
    There exists a tree $\cT_M = (V, E_{\cT_M})$ rooted at $s$ such that for every $u \in V \setminus \{s\}$ and every ancestor $a$ of $u$ in $\cT_M$, $\Near_{\cT_M}(u,a)=\Near_G(u,a)$ and $c(\Near_{\cT_M}(u,a))=c(\Near_G(u,a))$.
\end{theorem}
Any minimal Gomory-Hu Tree of $G$ stores at least one of $\Near(u,v)$ and $\Near(v,u)$ for every $u,v\in V$. 
This can be seen as follows.
If $u$ is an ancestor of $v$ in $\cT_M$ (or vice versa), then $\Near(u,v)$ is captured by $\cT_M$.
Otherwise, if $u,v$ are independent in $\cT_M$, then assume without loss of generality that $\lambda_{u,s}\le \lambda_{v,s}$.
Then, it is easy to show $\Near(u,s)=\Near(u,v)$.
Notice that no $u,v$-mincut can keep $u$ and $s$ on the same side of the cut, since $\lambda_{u,s}\le \lambda_{v,s}$.
Therefore, $\lambda_{u,v}\ge \lambda_{u,s}$.
Combine this with the fact that $\Near(u,s)$ separates $u,v$ we have $\lambda_{u,s}\ge \lambda_{u,v}$, and therefore we achieve equality $\lambda_{u,v}=\lambda_{u,s}$.
This implies that $\Near(u,s)$ is a $u,v$-mincut, and since it is the nearest $u,s$-mincut, it is also the nearest $u,v$-mincut.

While it seems natural to hope that one can encode all-pairs nearest mincuts using few minimal Gomory-Hu trees, we show that this is not the case.
The following theorem states this lower bound formally.
\begin{theorem}
    \label{theorem:gomory-hu-lower-bound}
    There exists an undirected weighted graph $G$ on $n$ vertices such that encoding $\Near(u,v)$ for every $u,v\in V$ 
    in $G$ requires $\Omega(n)$ minimal Gomory-Hu trees.  
\end{theorem}
\begin{figure}[h]
    \centering
    \scalebox{1.5}{\begin{tikzpicture}[
  v/.style   = {circle, draw, fill=white, inner sep=0pt, minimum size=7pt},
  wt/.style  = {font=\small, fill=white, inner sep=1pt},
  lbl/.style = {font=\small}
]

\node[v, label={[lbl]above:$a_1$}] (a1) at (0,    2) {};
\node[v, label={[lbl]left: $b_1$}] (b1) at (0,    1) {};
\node[v, label={[lbl]below:$c_1$}] (c1) at (-0.5, 0) {};
\node[v, label={[lbl]below:$d_1$}] (d1) at ( 0.5, 0) {};
\draw (a1) -- node[wt, right]       {$3$} (b1);
\draw (b1) -- node[wt, above left]  {$2$} (c1);
\draw (b1) -- node[wt, above right] {$1$}  (d1);
\draw (c1) -- node[wt, below]       {$1$}  (d1);

\node[v, label={[lbl]above:$a_2$}] (a2) at (2,   2) {};
\node[v, label={[lbl]left: $b_2$}] (b2) at (2,   1) {};
\node[v, label={[lbl]below:$c_2$}] (c2) at (1.5, 0) {};
\node[v, label={[lbl]below:$d_2$}] (d2) at (2.5, 0) {};
\draw (a2) -- node[wt, right]       {$3$} (b2);
\draw (b2) -- node[wt, above left]  {$2$} (c2);
\draw (b2) -- node[wt, above right] {$1$}  (d2);
\draw (c2) -- node[wt, below]       {$1$}  (d2);

\node[v, label={[lbl]above:$a_3$}] (a3) at (4,   2) {};
\node[v, label={[lbl]left: $b_3$}] (b3) at (4,   1) {};
\node[v, label={[lbl]below:$c_3$}] (c3) at (3.5, 0) {};
\node[v, label={[lbl]below:$d_3$}] (d3) at (4.5, 0) {};
\draw (a3) -- node[wt, right]       {$3$} (b3);
\draw (b3) -- node[wt, above left]  {$2$} (c3);
\draw (b3) -- node[wt, above right] {$1$}  (d3);
\draw (c3) -- node[wt, below]       {$1$}  (d3);

\node[v, label={[lbl]above:$a_n$}] (an) at (7,   2) {};
\node[v, label={[lbl]left: $b_n$}] (bn) at (7,   1) {};
\node[v, label={[lbl]below:$c_n$}] (cn) at (6.5, 0) {};
\node[v, label={[lbl]below:$d_n$}] (dn) at (7.5, 0) {};
\draw (an) -- node[wt, right]       {$3$} (bn);
\draw (bn) -- node[wt, above left]  {$2$} (cn);
\draw (bn) -- node[wt, above right] {$1$}  (dn);
\draw (cn) -- node[wt, below]       {$1$}  (dn);

\draw (a1) -- node[wt, above] {$3$} (a2);
\draw (a2) -- node[wt, above] {$3$} (a3);
\draw (a3) -- (4.8, 2);
\node[font=\small] at (5.5, 2) {$\cdots$};
\draw (6.2, 2) -- node[wt, above] {$3$} (an);

\end{tikzpicture}}
    \caption{Illustration of the graph $G$.}
    \label{fig:gomory-hu-lower-bound}
\end{figure}
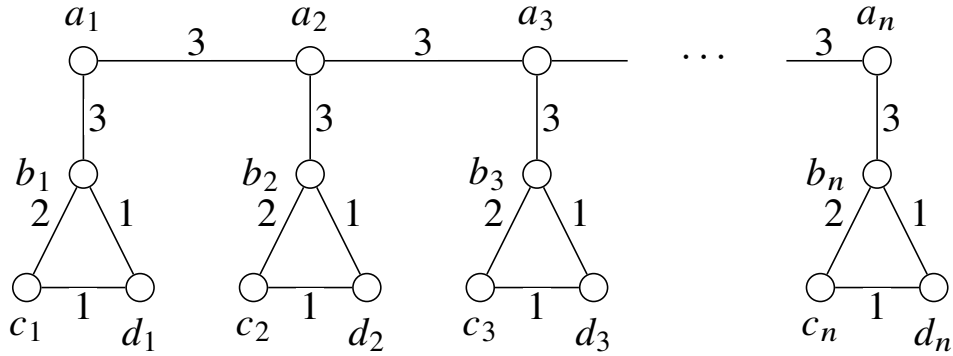
\begin{proof}
    The proof is based on a graph construction $G=(V,E)$ on $4n$ vertices, in which there exist $n$ unique cuts, each of which is only captured by a different minimal Gomory-Hu tree.
    Therefore, at least $n$ Gomory-Hu trees are required to represent all these cuts.
    Let $G$ be the graph defined as follows.
    Divide the vertices of $G$ into four sets: $A=\{a_1,\ldots,a_n\}$, $B=\{b_1,\ldots,b_n\}$, $C=\{c_1,\ldots,c_n\}$ and $D=\{d_1,\ldots,d_n\}$.
    The vertices in $A$ form a line, with edges of weight $3$ connecting $a_i$ to $a_{i+1}$ for every $i\in [n-1]$.
    The vertices $b_i,c_i,d_i$ form a triangle for every $i\in [n]$.
    Each vertex $b_i$ is connected to $a_i$ by an edge of weight $3$, each vertex $b_i$ is connected to $d_i$ by an edge of weight $2$.
    Finally, each vertex $c_i$ is connected to $b_i$ and to $d_i$ by edges of weight $1$.
    An illustration of the graph is provided in \Cref{fig:gomory-hu-lower-bound}.

    We now show that $\Near(b_i,c_i)$ is only represented in the minimal Gomory-Hu tree rooted at $c_i$ or $d_i$.
    This immediately implies the result since there are $n$ such cuts, and each of them is only represented in a different minimal Gomory-Hu tree.

    Recall that by the definition of the minimal Gomory-Hu tree, for every vertex $v$ whose parent is $p(v)$, the cut represented by the edge $(v,p(v))$ is $\Near(v,p(v))$.
    It is straightforward to see that $\Near(c_i,x)=\set{c_i}$ for every $x\in V\setminus \set{d_i}$.  
    Similarly, $\Near(d_i,x)=\set{d_i}$ for every $x\in V$.
    Therefore, the vertices $c_i,d_i$ are both leaves in every minimal Gomory-Hu tree that is not rooted at $c_i$ or $d_i$.
    In addition, $\Near(b_i,c_i)=V\setminus \set{c_i,d_i}$. 
    To conclude, observe that since $c_i,d_i$ are independent vertices in every minimal Gomory-Hu tree that is not rooted at $c_i$ or $d_i$, the cut $\Near(b_i,c_i)$ cannot be represented by any edge in such a tree.
\end{proof}

\section{Conclusion and Future Works}
\label{sec : conclusion}
In this work, we present the first optimal compact representation for the family of all-pairs nearest (and symmetrically, farthest) mincuts in undirected weighted graphs. 
Our data structure occupies $O(n)$ space while supporting worst-case optimal $O(n)$ time queries. 
It thus matches the bounds of the classical Gomory-Hu tree despite storing the substantially richer family of all-pairs nearest mincuts.
As an application, we obtain optimal-space insertion sensitivity oracles for all-pairs mincuts, improving previous quadratic-space bounds while still maintaining relatively efficient query times.

Beyond our algorithmic results, our work also exposes new structural properties of all-pairs nearest mincuts, including the chain structure underlying our compact representation, which we believe may be of independent interest.
At a technical level, our work builds upon the connectivity carcass of \cite{DBLP:conf/stoc/DinitzV94}, which was designed to encode Steiner mincuts.
Our work, combined with the recent work of \cite{BhanjaPP2026}, suggests that the carcass may be a useful tool for various cut problems.

\medskip
\noindent
\textbf{The edge-deletion regime.}
While our work resolves the space complexity of all-pairs sensitivity oracles for edge insertions, it leaves open the failure regime.
Existing sensitivity oracles for edge failures either use $O(\min\{n^2,m\})$ space \cite{DBLP:conf/soda/BaswanaP22}, or use $O(n)$ space but are restricted to the single-source setting \cite{BhanjaPP2026}.
We believe that an $O(n)$ space sensitivity oracle for all-pairs mincuts exists, and believe the nearest mincut hierarchy is a natural starting point towards this.

\medskip
\noindent
\textbf{Faster queries.}
The main shortcoming of our sensitivity oracle is that it requires $O(n)$ time to report whether the mincut value of a given pair of vertices increases after an edge insertion.
While this is worst-case optimal for reporting an entire nearest mincut, it is not optimal for simply reporting whether the mincut value has increased.
Previous work has given an optimal $O(1)$ time query structure~\cite{DBLP:conf/soda/BaswanaP22}, although at the cost of using $O(n^2)$ space. 
The natural question is whether one can achieve $o(n)$ query time while still using $O(n)$ space.

\medskip
\noindent
\textbf{Data structure construction.} 
Given the optimal space and query complexity of our nearest mincut hierarchy, improving its construction time is a natural next objective. It seems plausible to construct the nearest mincut hierarchy using ${O}(n\polylog{n})$ maxflow computations. 
However, it would be desirable to achieve $\polylog{n}$ maxflow computations, matching the bounds for fast Gomory-Hu tree constructions, e.g., \cite{DBLP:conf/focs/Abboud0PS23}.

\subsection*{AI Disclosure}
In the writing of this paper, the authors used GitHub Copilot for in-line writing suggestions. 
In addition, the authors used Claude for feedback on \Cref{sec:near-minimum-cut-query,sec:nearest-mincut-tree-construction-query}, particularly on the clarity and correctness of both the English prose and the mathematical content. 
All proof ideas and technical content were developed by the authors, who maintain full responsibility for the content of the paper.

{\small
\bibliographystyle{alphaurl}
\bibliography{references}
}

\appendix
\section{Missing Proofs}
\label{sec: missing proofs}
\begin{proof} [Proof of \Cref{lem:sub-posi-general}]
    We prove the first item with $S$ and $T$; the second item and the ``furthermore'' parts are similar.
    Let $C$ be the $S$-cut among $A \cap B$ or $A \cup B$, and $D$ be the other one which is a $T$-cut.
    As $A$ and $B$ are $S$-mincut and $T$-mincut, we get (i) $c(A) \leq c(C)$ and (ii) $c(B) \leq c(D)$. But by submodularity, $c(A) + c(B) \geq c(C) + c(D)$, so (i) and (ii) must in fact hold with equality, hence $C$ is an $S$-mincut and $D$ is a $T$-mincut.
\end{proof}
\begin{lemma} \label{lem: mincut does not separate any incomparable node}
Let $u,v,u',v' \in V$, with $\alpha = \LCA_\cT (u,v)$ and $\beta = \LCA_\cT (u',v')$.
If $\LCA_{\cT}(\alpha,\beta)\notin \{\alpha,\beta\}$, then there is no $u,v$-mincut which separates $u'$ from $v'$.
\end{lemma}
\begin{proof}
    Seeking contradiction, let $C$ be a $u,v$-mincut which separates $u',v'$, wlog $u,u' \in C$ and $v,v' \in \overline{C}$.
    Denote $\mu = \LCA_\cT (\alpha, \beta)$, and consider the children $\mu_\alpha$ and $\mu_\beta$ of $\mu$ which are ancestors of $\alpha$ and $\beta$, respectively.
    As $\mu_\alpha$ and $\mu_\beta$ are two different $\mu$-equivalence classes, there exists a $\mu$-mincut $U$ such that $u,v \in \alpha \subseteq \mu_\alpha \subseteq U$ and $u'
    ,v' \in \beta \subseteq \mu_\beta \subseteq \overline{U}$.
    So, $C \cap U$ separates $u,v$ and is therefore a $\alpha$-cut, and $C \cup U$ separates $u',v'$ and is therefore a $\mu$-cut.
    By submodularity (\Cref{lem:sub-posi-general}), we obtain that $C \cup U$ is a $\mu$-mincut which separates $u',v'$.
    But this is a contradiction, since $u',v'$ are in the same $\mu$-equivalence class $\mu_\beta$.
\end{proof}

\end{document}